\documentclass{llncs}
\usepackage{graphicx} 
\usepackage{algorithm}
\usepackage{algpseudocode}
\usepackage{amsmath}
\usepackage{amssymb}
\usepackage{pgfplots}
\usepackage{booktabs}
\usepackage{subcaption}
\usepackage{multirow}
\pgfplotsset{compat=1.18}
\usepackage{tikz}
\usetikzlibrary{arrows.meta, positioning}
\usetikzlibrary{trees, arrows.meta}
\usetikzlibrary{matrix, positioning, backgrounds, fit, arrows.meta}

\usepackage[hidelinks]{hyperref}

\usepackage{etoolbox}
\newtoggle{extended}
\toggletrue{extended}  
\newcommand{\pubext}[2]{%
  \iftoggle{extended}{#2}{#1}%
}

\usepackage{fancyvrb}
\fvset{fontsize=\footnotesize}

\usepackage{xcolor}
\usepackage{xspace}

\newcommand{\uep}{$uEP$\xspace}
\newcommand{\lep}{$\ell EP$\xspace}
\newcommand{\leep}{$\ell^*EP$\xspace}
\newcommand{\G}{\mathcal{G}}
\newcommand{\N}{\mathcal{N}}
\newcommand{\E}{\mathcal{E}}
\newcommand{\V}{\mathcal{V}}
\newcommand{\U}{\mathcal{U}}
\newcommand{\PP}{\mathcal{P}}
\newcommand{\LL}{\mathcal{L}}
\newcommand{\A}{\mathcal{A}}

\newcommand{\Lp}{L_{\texttt{OS\underline{P}}}}
\newcommand{\Ls}{L_{\texttt{PO\underline{S}}}}
\newcommand{\Lo}{L_{\texttt{SP\underline{O}}}}
\newcommand{\Ap}{A_{\texttt{P}}}
\newcommand{\As}{A_{\texttt{S}}}
\newcommand{\Ao}{A_{\texttt{O}}}

\newcommand{\leap}{\textit{leap}}
\newcommand{\timeout}{\textcolor{gray}{TO}}

\newcommand{\edgepattern}[3]{$\mathtt{(}$#1$\mathtt{)-[}$#2$\mathtt{]\rightarrow(}$#3$\mathtt{)}$}

\newcommand{\GN}[1]{\textcolor{blue}{GN: #1}}
\newcommand{\AH}[1]{\textcolor{green!50!black}{AH: #1}}

\newcommand{\ah}[1]{\textcolor{red!50!black}{#1}}

\newcommand{\camready}[1]{#1}

\newcommand{\no}[1]{}

\title{Uplifting the Superpowers \\ of Worst-Case-Optimal Join Algorithms\thanks{Funded by ANID -- Millennium Science Initiative Program -- Code ICN17\_002, Chile. AGB is funded by MCIN/AEI/10.13039/501100011033 and EU/ERDF ``A way of making Europe'' PID2022-141027NB-C21 (EARTHDL).  CITIC is funded by the Xunta de Galicia and co-financed by the EU through FEDER Galicia: grant ED431G 2023/01. GN is funded by Fondecyt Grant 1260080, ANID, Chile.}}
\author{Adrián Gómez-Brandón\inst{1,3} \and Aidan Hogan\inst{2,3} \and Gonzalo Navarro\inst{2,3}}
\institute{
Universidade da Coru\~na, CITIC, Facultade de Informática, Spain
\and
Department of Computer Science, University of Chile, Chile 
\and
IMFD --- Millennium Institute for Foundational Research on Data, Chile
}

\begin{document}

\maketitle

\begin{abstract}
Worst-case-optimal (wco) join algorithms have demonstrated their power -- in both theory and practice -- to efficiently solve complex Basic Graph Patterns (BGPs). Modern graph query languages, such as SPARQL and GQL, have BGPs at their core, but also have a wide range of other features, including filters (aka.\ selections). Such conditions are typically handled via pre- or post-filtering, before or after processing the BGPs. In this paper we show how to uplift wco join algorithms so as to incorporate such filtering natively, improving efficiency. We demonstrate the superiority of this approach by extending the \textit{Ring} -- a compact index that provides wco resolution of BGPs within almost no extra space on top of the graph -- so as to handle property graphs using our new techniques while retaining compactness. We implement this extension and experimentally show that it outperforms various baseline systems.

\end{abstract}

\section{Introduction}


There is an ever-growing demand for increasingly-efficient systems for querying knowledge graphs. In the context of open knowledge graphs, RDF and SPARQL dominate~\cite{MalyshevKGGB18}. In the context of enterprise knowledge graphs, property graphs, query languages such as Cypher~\cite{FrancisGGLLMPRS18} and GQL~\cite{DeutschFGHLLLMM22}, and commercial engines such as Neo4j~\cite{Webber12}, are a popular alternative. Despite these seemingly disparate graph models and query languages, the same  key primitives -- in particular, graph patterns~\cite{AABHRV17} -- underlie both RDF/SPARQL and property graph query engines. Techniques for querying one form of knowledge graph model can often be applied for the other, with various systems~\cite{BebeeCGGKKMMPRR18,LeeuwenMWFY22,VR+23} natively supporting both.

A promising recent line of research for efficiency evaluating graph queries is to combine compact data structures and worst-case optimal joins~\cite{arroyuelo2024ring,ArroyueloNRR22,BigerlCBSSN20,LiuLPLZ25,ArroyueloBFGN26}. Such works have focused on evaluating BGPs over triple-based models such as RDF. We extend these works in two ways. First, we show how techniques originally proposed for RDF-like models -- specifically the Ring~\cite{arroyuelo2024ring} -- can be adapted to the property graph model. Second, we further add support for selections, aka.\ \textit{filters}, which are among the most widely-used query operators in practice: Bonifati et al.~\cite{BonifatiMT20} find, for example, that 40\% of SPARQL queries in real-world logs use \texttt{FILTER}, second only to \texttt{SELECT}. While filters can be evaluated over the results of a BGP, one of the most fundamental query optimizations is to ``push down'' filters, using them to eliminate intermediate results as soon as possible. We propose techniques to enable the push-down of filters by translating them into efficient operations directly over a custom compact data structure.

\paragraph{Paper structure.} Section~\ref{sec:related} discusses related works on compact data structures and wco joins for graphs. Section~\ref{sec:pre} presents preliminaries. Section~\ref{sec:leaptors} introduces the core technical abstraction of ``leaptors'': iterators enabling wco execution for graph queries with filters. 
Section~\ref{sec:pg} proposes the PG-Ring for the concrete setting of property graphs and GQL-like queries, describing the implementation of a static, in-memory database engine using leaptors. Section~\ref{sec:exp} presents experimental results for two benchmarks comparing PG-Ring with internal and external baselines. Section~\ref{sec:conclude} wraps up our contributions and presents future work.

\section{Related Work}\label{sec:related}


In the context of efficiently evaluating graph patterns, recent years have seen two promising lines of research emerge.

The first line of research refers to \textit{worst-case optimal} (wco) joins~\cite{NgoPRR18}, which provide -- in the context of querying graphs -- theoretical and practical guarantees on the maximum cost of enumerating all solutions for a basic graph pattern (BGP). Such algorithms have been shown empirically to provide notable speed-ups over traditional join algorithms, upholding their theoretical underpinnings for RDF/SPARQL~\cite{HRRSiswc19,BigerlCBSSN20,VR+23} and similar triple-based graph models~\cite{MhedhbiS19}. However, the benefits of wco joins in terms of time often come at the cost of space~\cite{HRRSiswc19,MhedhbiS19,VR+23} as they often require storing additional index orders.

The second line of research explores compressed data structures for indexing and querying graphs. In the RDF/SPARQL world, formats such as HDT~\cite{FernandezMGPA13}, COTTAS~\cite{ArenasGuerreroF25}, etc., propose compressed RDF file formats that can evaluate triple patterns, which can be used to build more complex query engines~\cite{TaelmanHSV18}. 

More recent works have combined both lines of research by investigating succinct data structures that support wco joins. These in-memory structures include the Ring (based on wavelet trees)~\cite{arroyuelo2024ring}, QDags (based on quadtrees)~\cite{ArroyueloNRR22}, Hypertries (based on bit tensors)~\cite{BigerlCBSSN20}, RDF-TDAA (based on Directly Addressable Arrays)~\cite{LiuLPLZ25}, and RDFCSA (based on compressed suffix arrays)~\cite{ArroyueloBFGN26}. These works support wco queries while addressing their key limitation: space. 

However, to the best of our knowledge, no work has investigated succinct data structures that support wco joins and filters (beyond equality). Furthermore, we are not aware of works applying wco joins over property graphs.

\section{Graph Databases and Worst-Case-Optimal Algorithms}
\label{sec:pre}

\subsection{RDF and Basic Graph Patterns}

An RDF graph \cite{rdfconcepts11} $G$ is a set of triples $(s,p,o)$, each denoting an edge $s \stackrel{p}{\to} o$. Node $s$ is called the {\em subject}, label $p$ the {\em predicate}, and node $o$ the {\em object}. A {\em Basic Graph Pattern (BGP)} $Q$ is a set of {\em triple patterns} $(x,y,z)$ where any element can be a constant or a variable. A {\em solution} to a BGP is a function that assigns a constant to each variable of the BGP so that all resulting triples occur in the graph. {\em Solving} a BGP $Q$ consists of producing the set $Q(G)$ of all its solutions.

Let $G$ have $N$ triples. The {\em AGM bound} \cite{AGM13} of $Q$, $Q^* \ge |Q(G)|$, is the maximum size of a solution set of $Q$ over any graph of $N$ triples. An algorithm solving $Q$ is {\em worst-case-optimal} (wco)~\cite{AGM13,HRRSiswc19} if it works in time $\tilde{O}(Q^*)$, where $\tilde{O}$ hides factors polylogarithmic on $N$, or independent of $N$ (such as $|Q|$). The wco concept and wco algorithms have had notable impact on the resolution of complex and cyclic BGPs. Such algorithms outperform every algorithm that enforces the triple pattern conditions {\em consecutively} (as classic database join plans do) by enforcing all the conditions {\em simultaneously}. The classical example is the triangle query $R(a,b) \bowtie S(b,c) \bowtie T(c,a)$, which if the tables have $N$ pairs takes $\Omega(N^2)$ worst-case time with any consecutive-enforcing algorithm, whereas wco algorithms solve it in time $\tilde{O}(N^{3/2})$.

\subsection{Leapfrog TrieJoin}

{\em Leapfrog Triejoin (LTJ)} \cite{leapfrog} is a seminal wco join algorithm. It precalculates six tries where the graph triples are inserted as strings of length 3 in all different orders, $\mathcal{T} = \{ \texttt{SPO}, \texttt{SOP}, \texttt{POS}, \texttt{PSO}, \texttt{OSP}, \texttt{OPS}\}$. It chooses an order to handle the variables, which is called a {\em variable elimination order (VEO)}, and then assigns each triple pattern $t \in Q$ to a trie $\tau_t \in \mathcal{T}$, such that (i) the constants of $t$ appear first in the order of $\tau_t$, and (ii) the variables of $t$ appear in $\tau_t$ in the same order of the VEO. For each triple pattern $t \in Q$, LTJ first descends by its constants on $\tau_t$, and the node $v_t \in \tau_t$ arrived at is the {\em locus} of $t$. Next, LTJ takes one variable $x$ after the other in the VEO. Let $T_x = \{ t \in Q, \textrm{ $x$ appears in } t\}$. For each $t \in T_x$, the possible values of $x$ are the children of $v_t$. LTJ then {\em intersects} the children of all the loci $\{ v_t, t \in T_x\}$. For each value $c$ in the intersection, it branches with the {\em binding} $x := c$ and descends by value $c$ from all the involved loci, $v_t \gets child(v_t,c)$ for all $v_t \in T_x$. By the time all the variables are bound in this way, the set of bindings in that branch forms a new solution to report. Fig.~\ref{fig:tries} shows a toy graph and its trie \texttt{POS} (ignore property \textsf{age} of nodes for now).

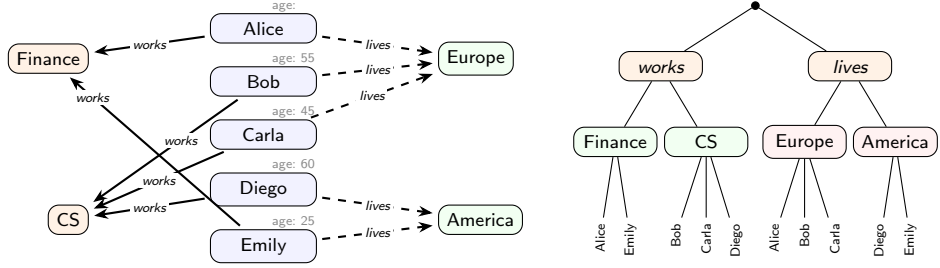
\begin{figure}[t]
\begin{tikzpicture}[
    font=\sffamily\scriptsize,
    arrow/.style={-{Stealth[scale=0.8]}, thick, shorten >=2pt, shorten <=2pt},
    label_style/.style={font=\tiny\sffamily, fill=white, inner sep=1pt},
    age_style/.style={font=\tiny\sffamily, text=gray!90, anchor=south east, inner sep=0.5pt, yshift=0pt}
]

    \begin{scope}[node distance=0.3cm and 1.6cm, yshift=-1.5cm]
        \node[rectangle, draw, fill=blue!5, rounded corners, minimum width=1.4cm] (Alice) {Alice};
        \node[age_style] at (Alice.north east) {age: \ \ \ \ }; 
        
        \node[rectangle, draw, fill=blue!5, rounded corners, minimum width=1.4cm, below=of Alice] (Bob) {Bob};
        \node[age_style] at (Bob.north east) {age: 55};
        
        \node[rectangle, draw, fill=blue!5, rounded corners, minimum width=1.4cm, below=of Bob] (Carla) {Carla};
        \node[age_style] at (Carla.north east) {age: 45};
        
        \node[rectangle, draw, fill=blue!5, rounded corners, minimum width=1.4cm, below=of Carla] (Diego) {Diego};
        \node[age_style] at (Diego.north east) {age: 60};
        
        \node[rectangle, draw, fill=blue!5, rounded corners, minimum width=1.4cm, below=of Diego] (Emily) {Emily};
        \node[age_style] at (Emily.north east) {age: 25};

        \node[rectangle, draw, fill=orange!10, rounded corners, left=of Alice, yshift=-0.4cm] (Finance) {Finance};
        \node[rectangle, draw, fill=orange!10, rounded corners, left=of Diego, yshift=-0.4cm] (CS) {CS};
        \node[rectangle, draw, fill=green!5, rounded corners, right=of Alice, yshift=-0.4cm] (Europe) {Europe};
        \node[rectangle, draw, fill=green!5, rounded corners, right=of Diego, yshift=-0.4cm] (America) {America};

        \draw[arrow] (Alice) -- (Finance) node[midway, label_style] {{\em works}};
        \draw[arrow] (Emily) -- (Finance) node[label_style, pos=0.85] {\em works};
        \draw[arrow] (Bob) -- (CS) node[midway, label_style, pos=0.4] {\em works};
        \draw[arrow] (Carla) -- (CS) node[midway, label_style] {\em works};
        \draw[arrow] (Diego) -- (CS) node[midway, label_style] {\em works};

        \draw[arrow, dashed] (Alice) -- (Europe) node[midway, label_style] {\em lives};
        \draw[arrow, dashed] (Bob) -- (Europe) node[midway, label_style] {\em lives};
        \draw[arrow, dashed] (Carla) -- (Europe) node[midway, label_style] {\em lives};
        \draw[arrow, dashed] (Diego) -- (America) node[midway, label_style] {\em lives};
        \draw[arrow, dashed] (Emily) -- (America) node[midway, label_style] {\em lives};
    \end{scope}

    \begin{scope}[xshift=6.5cm, yshift=-1.2cm,
        level 1/.style={sibling distance=2.5cm, level distance=0.8cm},
        level 2/.style={sibling distance=1.2cm, level distance=1cm},
        level 3/.style={sibling distance=0.4cm, level distance=1cm},
        node_box/.style={rectangle, draw, rounded corners, fill=gray!5, minimum width=1.1cm, font=\sffamily\scriptsize},
        leaf/.style={rotate=90, anchor=east, font=\tiny\sffamily, inner sep=2pt}
    ]
        \node[circle, draw, fill=black, inner sep=1pt] {}
            child { node[node_box, fill=orange!10] {\em works}
                child { node[node_box, fill=green!5] {Finance}
                    child { node[leaf] {Alice} }
                    child { node[leaf] {Emily} }
                }
                child { node[node_box, fill=green!5] {CS}
                    child { node[leaf] {Bob} }
                    child { node[leaf] {Carla} }
                    child { node[leaf] {Diego} }
                }
            }
            child { node[node_box, fill=orange!10] {\em lives}
                child { node[node_box, fill=red!5] {Europe}
                    child { node[leaf] {Alice} }
                    child { node[leaf] {Bob} }
                    child { node[leaf] {Carla} }
                }
                child { node[node_box, fill=red!5] {America}
                    child { node[leaf] {Diego} }
                    child { node[leaf] {Emily} }
                }
            };
    \end{scope}

\end{tikzpicture}
\caption{On the left, a graph database of where people live and area they work in; person nodes have also property \textsf{age}. On the right, the trie \texttt{POS} for the graph.}
\label{fig:tries}
\end{figure}

A way to implement LTJ intersections is to take the smallest child value $c$ of one locus $v_t \in T_x$, and search the children of the next locus $v_{t'} \in T_x$ for the smallest child value $c' \ge c$; this operation is called {\em leap}, that is, $c' \gets \leap(c)$. We then assign $c \gets c'$ and go on with the third locus $v_{t''} \in T_x$, and so on. We continue cycling over the loci until a whole pass over $T_x$ retains the same $c$ value, so $c$ is the next element in the intersection. After branching on $x := c$, we restart the process with the next child value in some loci. This algorithm is optimal in the sense that its cost is $\tilde{O}(alt)$, where $alt$ is the {\em alternation complexity}~\cite{BK03} of the sets to intersect. This is the least number of switches from one sorted sequence of values to another needed to cover them all; for example intersecting $\{ 2i, 1 \le i \le n \} \cap \{ 2i-1, 1 \le i \le n\} = \emptyset$ is hard because we must switch at every new value ($alt = 2n-1$), whereas intersecting $\{ i, 1 \le i \le n\} \cap \{ i, n < i \le 2n\} = \emptyset$ is easy ($alt=1$; it suffices to compare two elements to certify the output). Other intersection techniques used in some LTJ implementations, like probing the elements of the smallest set in the others, do not yield this guarantee. 

\subsection{Beyond RDF: Filtering by properties and labels}
\label{sec:beyond-rdf}

Most graph query languages, like SPARQL \cite{sparql11}, GQL \cite{francis2023researcher}, and SQL/PGQ \cite{DeutschFGHLLLMM22}, have BGPs at their core, but also include filters. GQL and SQL/PGQ assume property graphs \cite{angles2018property}, which extend triples to allow nodes and edges to have {\em properties} with {\em values}. Queries may then specify that, to bind some variable $x$ to a node or edge in the triple pattern, that node or edge must contain a property $\pi$ and that its corresponding value $x.\pi$ must be within a range, say $a \le x.\pi \le b$ (where $a$ and $b$ can be constants or other property values, say $x.\pi \le y.\pi'$). It is also customary to assign {\em labels} to nodes and edges, so that $x.\ell$ requires that variable $x$ is assigned to an element having label $\ell$. While SPARQL filters are applied directly on nodes and predicates in triples, we will allow zero-to-many values per property and zero-to-many labels, and hence our scheme is also flexible enough to capture features of RDF, as we will discuss in Section~\ref{sec:ext-ring}.

While the BGP core of the query can be solved efficiently with a wco algorithm, conditions on properties and labels have been typically handled outside of this wco algorithm by pre- or post-filtering: 
we can either filter all nodes and edge values $v$ to leave for consideration only those holding $a \le v.\pi \le b$ (or where label $v.\ell$ exists), or first solve the BGP and filter every output, removing bindings of $v$ that do not satisfy the conditions. If we use LTJ, a better solution is to adapt it so as to post-filter the values $v$ as soon as they come out of the intersection of loci children, thereby pruning branches earlier.

However, as with non-wco algorithms, these approaches may force us to examine many more solutions than necessary. Consider, for the graph of Fig.~\ref{fig:tries}, the BGP $\{ (x,\textsf{\em lives},\textsf{Europe}), (x,\textsf{\em works},\textsf{CS}) \}$ complemented with the condition $x.\textsf{age} \ge 50$, aimed at finding senior people living in Europe and working in CS. LTJ will choose two loci, say in \texttt{POS} (it could also be \texttt{OPS}), by descending through $(\textsf{\em lives},\textsf{Europe})$ and $(\textsf{\em works},\textsf{CS})$, respectively. The intersection of the loci children yields all people $x$ in Europe working in CS. If we pre-filter by age on a large real graph of this form, we will generate a lot of senior people, very few of which will end up passing the BGP filter. If we post-filter, we will generate a lot of people in Europe working in CS, very few of which will pass the seniority filter. None of those methods meet the alternation complexity bound, which is never asymptotically larger than the minimum of both sets, and can be much smaller.

\section{Leaptors: Incorporating Filters into LTJ} \label{sec:leaptors}

We generalize the logic of LTJ so as to handle not only triple patterns mapped to trie loci, but in general {\em leaptors} (these extend ``iterators'', which can only deliver the elements sequentially). A leaptor represents a set $S$ with a total order and supports operation {\em leap$(c) = \min \{ c' \in S \cup \{+\infty\},~ c' \ge c\}$}, in time polylogarithmic on $|S|$. An example leaptor is the algorithm that finds the child of a trie node (the locus) with the smallest value $c'\ge c$; this can be implemented, say, via exponential search on the increasing sequence of child values.

\subsection{Properties and ranges of values}

We now define a new leaptor to incorporate filters on property values. Let $\U$ be the set of node identifiers in the graph (we defer properties and labels on edges to Section~\ref{sec:pg}, though they could be supported similarly), and $\A_{\pi}$ be the set of values a property $\pi$ takes in the graph. Assume that both sets are integers; if not, we map them to integer ranges $[1,|\U|]$ and $[1,|\A_{\pi}|]$ while preserving order. For example, by storing an array of all the node values, and another of all the values in $\A_{\pi}$, we map actual values to integers with a binary search (in $O(\log N)$ time, only once per query), and can map back any result in constant time. 

We now define a {\em discrete grid} $M_{\pi}$ of $[1,|\U|] \times [1,|\A_{\pi}|]$, and set a point at $(v,a)$ in $M_{\pi}$ for every node $v$ having property $\pi$ with value $a$. This supports $v$ having zero-to-many values for a property $\pi$. Query {\em leap$(c)$} on the leaptor for $a \le x.\pi \le b$ is implemented on $M_{\pi}$ via the orthogonal range query {\em leftmost$([c,+\infty],[a^+,b^-])$}, which finds a point with minimum $x$-value $c' \ge c$ whose $y$-value is within $[a^+,b^-]$ ({\em leap$(c)$} returns $+\infty$ if this range is empty). Here, $a^+$ is the smallest value $a^+ \ge a$ in $\A_{\pi}$ and $b^-$ is the largest value $b^- \le b$ in $\A_{\pi}$; both are found in $O(\log |\A_{\pi}|) \subseteq O(\log N)$ time via binary searches (done only once per query). The geometric query itself is called an {\em orthogonal range successor query}, and can be solved in time $O(\log^\epsilon N)$ for any constant $\epsilon>0$, using space linear in the number of points in $M_{\pi}$ \cite{NNswat12}. This space is proportional to the database size, which includes one unit of space per value each property takes for each node. The leaptor time is then always in $O(\log N)$. See Fig.~\ref{fig:grid}.

\begin{figure}[t]
\centering

\begin{tikzpicture}
\begin{axis}[
    width=6.8cm,
    height=5cm,
    symbolic x coords={Alice, Bob, Carla, Diego, Emily},
    xmin=Alice, xmax=Emily,
    enlarge x limits={abs=0.6cm}, 
    xtick={Alice, Bob, Carla, Diego, Emily},
    ymin=0, ymax=70,
    ylabel={Age},
    xlabel={Nodes},
    grid=both,
    grid style={dashed, gray!15},
    tick label style={font=\footnotesize\sffamily},
    label style={font=\footnotesize\bfseries},
    scatter,
    only marks,
    mark=*,
    mark size=1.8pt,
    mark options={fill=blue!70},
    axis on top
]

\fill[orange!15] (axis cs:Carla, 50) rectangle ({rel axis cs:1,0} |- {axis cs:Emily,70});

\addplot[blue!70, nodes near coords, every node near coord/.append style={font=\tiny, yshift=1pt, gray}] 
coordinates {
    (Bob, 55)
    (Carla, 45)
    (Diego, 60)
    (Emily, 25)
};
\end{axis}
\end{tikzpicture}

\caption{The grid $M_{\mathsf{age}}$ for the graph of Fig.~\ref{fig:tries} (only showing person nodes). The shadowed area corresponds to the query $\mathit{leftmost}([\textsf{Carla},+\infty] \times [50,+\infty])$.}
\label{fig:grid}
\end{figure}
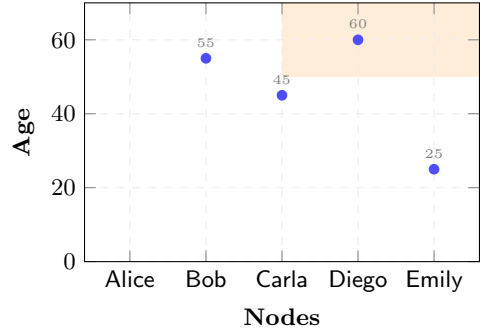

Let us return to our example of Section~\ref{sec:beyond-rdf} to illustrate the power of our method. Apart from the two leaptors implemented on tries for the triple patterns, we have the one implemented with grids for the condition $x.\textsf{age} \ge 50$. Once a candidate $c$ passes the filters of the two loci, we find the leftmost point $(c',a)$ in $[c,+\infty] \times [a^+,b^-]$. This will yield the next senior person $c' \ge c$, and then the next person in Europe working in CS will be sought after $c'$, not after $c$. In the example, the first candidate coming out of the intersection of children of trie loci $(\textsf{\em lives},\textsf{Europe})$ and $(\textsf{\em works},\textsf{CS})$ is \textsf{Bob}. The grid leaptor also returns \textsf{Bob} for the query $\leap(\mathsf{Bob}) = \mathit{leftmost}([\mathsf{Bob},+\infty]\times[50,+\infty])$, and thus \textsf{Bob} is an output of the query. If we now start from the grid to get the next candidate, its leaptor returns \textsf{Diego} (see the shaded area in Fig.~\ref{fig:grid}), which lets us discard \textsf{Carla} already. 

This process achieves the alternation complexity, visiting fewer candidates than pre- or post-filtering, and potentially not fully traversing either set (those passing the BGP and those passing the age filter). For example, if the person ids happen to be clustered by country, occupation, or age, the alternation complexity can be much lower than the number of candidates processed by standard filtering.

\subsection{Labels} \label{sec:labels}

We also define leaptors for labels. While labels can be handled as a special property (giving them value $0$ when the label exists, and asking for $\mathit{leftmost}([c,+\infty]\times[0,0])$ on their grid), a faster leaptor uses a {\em successor} data structure \cite{PT06} storing the set $D_\ell$ of node ids having label $\ell$. The structure can find the smallest integer $\ge c$ in $D_\ell$ (or return $+\infty$ if there is none) in time $O(\log\log |\U|)$, using $O(|D_\ell|)$ space. The leaptor for label $\ell$ then simply implements $\leap(c) = \mathit{successor}(D_\ell,c)$. 

It is also generally possible to ask for nodes having a property $\pi$, independently of its value. The corresponding leaptor can, again, be simulated with query $\mathit{leftmost}([c,+\infty]\times[-\infty,+\infty])$ on $M_\pi$, but within linear space we can also store a successor data structure on the set $D_\pi$ of the node ids having some value for property $\pi$, and run $\leap(c)$ in time $O(\log\log |\U|)$.

\subsection{ERing: A succinct wco index handling properties and labels}
\label{sec:ext-ring}

The Ring \cite{arroyuelo2024ring} is a {\em succinct} index that handles BGPs in wco time on RDF graphs. An index is said to be succinct if, including the data, it uses $\mathcal{S}+o(\mathcal{S})$ space, where $\mathcal{S}$ is the size of a plain representation of the data. The Ring represents the set of triples $(s,p,o)$ with three sequences, $\Lo$, $\Lp$, and $\Ls$. Sequence $\Lo$ stores the objects of the $N$ triples sorted by the order \texttt{SPO}, that is, triples $(s,p,o)$ are sorted lexicographically and $\Lo$ collects the sequence of their $o$ components. Analogously, sequence $\Lp$ collects the $N$ predicates $p$ in the order \texttt{OSP} and sequence $\Ls$ collects the $N$ subjects $s$ in the order \texttt{POS}. Overall, letting $\U$ be the universe of nodes and $\LL$ that of labels (or predicates), the Ring uses $(2N\log_2|\U|+N\log_2|\LL|) (1+o(1))$ bits, which is the space needed to store the $N$ triples $(s,p,o)$ in plain form, plus a sublinear redundant space. See Fig.~\ref{fig:ring}.

The sequences $L_*[1,N]$ are represented with a data structure called a {\em wavelet tree} \cite{GGV03,Nav14}, which on a universe $\Sigma$ (where $\Sigma$ is $\U$ or $\LL$) uses $N\log_2|\Sigma|(1+o(1))$ bits and in particular implements the operation $rank_c(L_*,i)$, the number of times $c$ occurs in $L_*[1,i]$, in time $O(\log|\Sigma|)$. We also represent arrays $\Ao$, $\Ap$, and $\As$, where $A_*[1,|\Sigma|]$ holds in $A_*[e]$ the number of positions in $L_*$ with values less than $e$. Those arrays can be represented with bitvectors $B_*[1,N]$ where $B_*[A_*[e]]=1$ for all $e$ and $0$ otherwise, so $A_*[e]=select_1(B_*,e)$, where  $select_b(B_*,e)$ is the position of the $e$th bit $b$ in $B_*$. Bitvector representations solving $select_b$ in constant time \cite{Cla96,Mun96} require $N+o(N)$ bits of space, which fits within the sublinear extra space of the Ring.

\begin{figure}[t]
\centering
\begin{tikzpicture}[
    font=\sffamily\fontsize{7.5pt}{8pt}\selectfont,
    array_style/.style={
        matrix of nodes,
        nodes={
            draw=none, 
            minimum width=1.05cm, 
            minimum height=0.35cm, 
            inner sep=0pt, 
            anchor=center
        },
        column sep=1pt,
        row sep=0pt,
        nodes in empty cells
    },
    header_l/.style={font=\bfseries\tiny, text=blue!80!black},
    connection/.style={-{Stealth[scale=0.8]}, thick, red!70, opacity=0.7, bend left=12},
    highlight/.style={draw=orange, thick, rounded corners=2pt, inner sep=0pt}
]

    \matrix (SPO) [array_style] {
        Alice & \itshape lives & Europe \\ 
        Alice & \itshape works & Finance \\
        Bob   & \itshape lives & Europe \\
        Bob   & \itshape works & CS \\
        Carla & \itshape lives & Europe \\
        Carla & \itshape works & CS \\
        Diego & \itshape lives & America \\
        Diego & \itshape works & CS \\
        Emily & \itshape lives & America \\
        Emily & \itshape works & Finance \\
    };
    \node[above=0.1cm of SPO] {\texttt{SPO} order};
    \node[header_l, above=0.05cm of SPO-1-3] {$\Lo$};

    \matrix (OSP) [array_style, right=0.6cm of SPO] {
        America & Diego & \itshape lives \\
        America & Emily & \itshape lives \\
        CS      & Bob   & \itshape works \\ 
        CS      & Carla & \itshape works \\
        CS      & Diego & \itshape works \\ 
        Europe  & Alice & \itshape lives \\ 
        Europe  & Bob   & \itshape lives \\
        Europe  & Carla & \itshape lives \\
        Finance & Alice & \itshape works \\
        Finance & Emily & \itshape works \\
    };
    \node[above=0.1cm of OSP] {\texttt{OSP} order};
    \node[header_l, above=0.05cm of OSP-1-3] {$\Lp$};

    \matrix (POS) [array_style, right=0.6cm of OSP] {
        \itshape lives & America & Diego \\
        \itshape lives & America & Emily \\
        \itshape lives & Europe  & Alice \\ 
        \itshape lives & Europe  & Bob \\
        \itshape lives & Europe  & Carla \\
        \itshape works & CS      & Bob \\   
        \itshape works & CS      & Carla \\
        \itshape works & CS      & Diego \\
        \itshape works & Finance & Alice \\
        \itshape works & Finance & Emily \\ 
    };
    \node[above=0.1cm of POS] {\texttt{POS} order};
    \node[header_l, above=0.05cm of POS-1-3] {$\Ls$};

    \begin{scope}[on background layer]
        \foreach \m in {SPO, OSP, POS} {
            \fill[blue!8] (\m-1-3.north west) rectangle (\m-10-3.south east);
        }
        
        \node[highlight, fit=(SPO-6-1) (SPO-6-3)] {};
        
        \node[highlight, fit=(POS-6-1) (POS-10-3)] {};
    \end{scope}

    \draw[connection] (SPO-1-3.east) to (OSP-6-1.west);
    \draw[connection] (OSP-6-3.east) to (POS-3-1.west);

\end{tikzpicture}

\caption{A Ring structure on the graph of Fig.~\ref{fig:tries}, assuming the node id order is the lexicographic order of the strings. The Ring only stores the three shaded sequences. The arrows show how the triple $(\texttt{Alice},\textsf{\em lives},\texttt{Europe})$ is tracked across the three orders. The rectangle on \texttt{POS}, corresponding to the child of the trie root by \textsf{\em works}, descends by \texttt{Carla} to the rectangle on \texttt{SPO}.}
\label{fig:ring}
\end{figure}

With those structures we can {\em track} triples across the arrays $L_*$: the triple at $\Lo[i]$ has object $o=\Lo[i]$, and it corresponds to $\Lp[i']$, where $i' = \Ao[o] + rank_o(\Lo,i)$. The triple then has predicate $p = \Lp[i']$, and it corresponds to $\Ls[i'']$, where $i'' = \Ap[p] + rank_p(\Lp,i')$. The subject of the triple is then $s=\Ls[i'']$ and if we compute $\As[s] + rank_s(\Ls,i'')$ we return to $\Lo[i]$.

Trie nodes correspond to ranges in some sequence $L_*$. A locus in \texttt{SOP} or \texttt{SPO}, reached by descending by $s$, corresponds to the range $\Lo[i,j] = \Lo[\As[s]+1,\As[s+1]]$ of the triples $(s,*,*)$. If we now descend by $o$ in \texttt{SOP}, the resulting locus corresponds to range $(s,o,*)$ in the order \texttt{SOP}, which we do not represent, but also to range $(o,s,*)$ in the order \texttt{OSP}, which we do represent in $\Lp$. This is $\Lp[i',j']$, where $i'=\Ao[o]+rank_o(\Lo,i-1)+1$ and $j' = \Ao[o] + rank_o(\Lo,j)$. If, instead, we descend by $p$ after $s$ in \texttt{SPO}, we stay in the same sequence $\Lo$, which represents that trie, now restricting the range to the triples $(s,p,*)$. The new range is computed by descending from the root of \texttt{POS} with $p$, and from the resulting range $\Ls[i'',j'']$ descending by $s$ to $\Lo[i''',j''']$, with $i'''=\As[s]+rank_s(\Ls,i''-1)+1$ and $j''' = \As[s] + rank_s(\Ls,j'')$. We will still call those operations ``descending'' by a value, even if they are mapped to ranges in $L_*$. See again Fig.~\ref{fig:ring} for examples of tracking and descending.

Descending operations allow the Ring to simulate the LTJ algorithm by mapping loci to ranges in $L_*$. The Ring leaptor is implemented as follows. Say we are again in $\Lo[i,j]$ after descending by $s$ in trie \texttt{SOP}. To implement {\em leap$(c)$} on the values of $o$, we must find the smallest value in $\Lo[i,j]$ that is $\ge c$. This operation, called {\em range next value}, is supported in $O(\log|\Sigma|)$ time by the wavelet tree \cite{GNPtcs11}. If, instead, we are in trie \texttt{SPO} and want to find the smallest value of $p$ that is $\ge c$, we start at range $\Ls[\Ap[c]+1,N]$ (the range of all predicates $\ge p$ in order \texttt{POS}), and compute $\Lo[i'] = \As[s] + rank_s(\Ls,\Ap[c])+1$. Tracking $\Lo[i']$ to $\Lp$ yields the desired value of $p$. The Ring then implements the leaptors in time $O(\log|\Sigma|)$, yielding a total query time complexity of $O(Q^* \cdot |Q|\log|\Sigma|) \subseteq \tilde{O}(Q^*)$.

\paragraph{Incorporating properties and labels.}
We add to the Ring succinct representations of grids $M_{\pi}$ that use $N_{\pi} \log_2 |\A_{\pi}| (1+o(1))$ bits, where $N_{\pi}$ is the sum of the number of values property $\pi$ has over all nodes.  This is again the size of the data plus a sublinear term. Such a succinct representation is the wavelet tree of the sequence $S_\pi[1,N_\pi]$ of values property $\pi$ takes, sorted by increasing identifiers of the nodes. We also represent a bitvector $B_{\pi}[|\U|+N_{\pi}]$ where, for each node identifier $v$, we append a $1$ and then as many $0$s as values $v$ has for property $\pi$; note that positions in $S_\pi$ align with $0$s in $B_\pi$. The grid of Fig.~\ref{fig:grid}, for example, would be represented with the sequence $S_{\textsf{age}} = \langle 55,45,60,25\rangle$ and the bitvector $B_{\textsf{age}}= 110101010$ (considering only the node ids that appear in the figure). 

The bitvector representation of $B_\pi$ \cite{Cla96,Mun96}  uses $(|\U|+N_\pi)(1+o(1))$ bits of space (thus fitting within the sublinear extra space of the index) and supports queries $select_b$ in constant time. 
We thus compute with $c' \gets select_1(B_{\pi},c)-c+1$ the first position in $S_\pi$ of node identifiers $\ge c$. The wavelet tree then solves the orthogonal range successor query from $S_\pi[c'..]$, in time $O(\log|\U|)$ \cite{BCNic13}. The answer $c''$ is finally converted back to an identifier with $select_0(B_\pi,c'')-c''$. 
In our example, to compute $\mathit{leftmost}([\textsf{Carla},+\infty] \times [50,+\infty])$ on $M_{\textsf{age}}$, where $\textsf{Carla}=3$, we start with $c' \gets select_1(B_{\textsf{age}},3)-3+1=2$, so we find the first value $\ge 50$ on $S_{\textsf{age}}[2..]$. This is at $S_{\textsf{age}}[3]=60$. We finally convert that position, $c''=3$, back to the identifier $select_0(B_{\textsf{age}},3)-3 = 4=\textsf{Diego}$.

Queries on labels were handled with generic successor data structures in Section~\ref{sec:labels}. While using linear space, this representation is not succinct: If a label $\ell$ is assigned to $m$ elements out of $n$, a succinct representation uses space close to the bits required to specify those $m$ elements, that is, $\log_2 {n \choose m} = m\log\frac{n}{m}+O(m)$ bits. To match that space, we store one bitvector $B_\ell$ for each label $\ell$, with one bit per node, marking with $1$ the nodes having label $\ell$. A {\em sparse} bitvector representation \cite[Sec.~4.4]{Nav16}, if $m$ out of $n$ bits are $1$, uses $m\log\frac{n}{m}+2m$ bits, and supports $rank_b$ and $select_0$ in time $O(\log m)$ and $select_1$ in constant time. Operation {\em leap$(c)$} for the condition $x.\ell$ is then implemented as $select_1(B_\ell,rank_1(B_\ell,c-1)+1)$ with this leaptor, in time $O(\log m)$ (which is logarithmic on the data size, and thus $\tilde{O}(1)$). This representation also provides a leaptor for the next element {\em not} having a label $\ell$, with $\leap(c) = select_0(B_\ell,rank_0(B_\ell,c-1)+1)$, also in time $O(\log m)$.

\paragraph{The ERing.}
The resulting structure, which we call {\em Extended Ring (ERing)}, is succinct on the triples, properties and labels it represents. It implements LTJ and supports leaptors for filtering by the existence of node properties and labels, and by restricting the values of node properties to ranges. The range limits can be constants or other property values. When a range involves several variables, as in $x.\textsf{age} \le y.\textsf{age}$, the leaptor is enabled only when binding the last variable ($x$ or $y$ in our example); by then, the condition has become a range with constants.

\subsection{Worst-case optimality}

The concept of worst-case optimality for BGPs (i.e., intersections of triple patterns) is well-defined for the RDF model. Node properties $\pi$ can be expressed in RDF by creating a node for each literal value $a \in \A = \cup_\pi \A_\pi$, and represent $v.\pi = a$ with a triple $(v,\pi,a)$. For labels we can use a reserved property \textsf{\em lab} and represent $v.\ell$ with $(v,\textsf{\em lab},\ell)$. Conditions like $x.\pi = a$ or $x.\ell$ are then expressed as triple patterns $(x,\pi,a)$ or $(x,\textsf{\em lab},\ell)$. Range conditions $a \le x.\pi \le b$ can be translated into $\{ (x,\pi,z), R(a,z), R(z,b)\}$ with variables $x$ and $z$. The (virtual) table $R(i,j)$ contains all the pairs $i,j \in \A$ such that $i \le j$. The range condition is then simulated by making node ids respect the order of the literals in $\A$ and using a simple leaptor for $z$: $\leap(c) = \max(a,c)$ if $c \le b$ and $+\infty$ otherwise. 

Mapping properties and labels into plain RDF and using the standard Ring~\cite{arroyuelo2024ring} is then wco, but not succinct: it creates one triple per node label or property value. We show that our ERing, instead, is a succinct-space wco solution.

\begin{definition}
A solution to BGPs extended with clauses for the existence of properties and labels, and comparing property values, is {\em worst-case optimal} if its time complexity is within the AGM bound of the corresponding query mapped onto the RDF model, as described above.
\end{definition}

\begin{theorem}
The ERing is worst-case-optimal when solving extended BGPs.
\end{theorem}
\begin{proof}
It suffices to observe that the ERing's time complexity is not worse than the standard Ring's \cite{arroyuelo2024ring} on the corresponding RDF graph, with a certain VEO: for each range condition $a \le x.\pi \le b$, translated into $\{ (x,\pi,z), R(a,z), R(z,b)\}$, we bind $z$ immediately after $x$. In the ERing, the potential values of $x$ and $z$ are filtered simultaneously, which per the alternation complexity is never worse than filtering one after the other. With the standard Ring, each potential value of $x$ will be instantiated and only then filtered by $z$. The conditions on the existence of properties $x.\pi$ are particular cases, $-\infty \le x.\pi \le +\infty$.

Simpler conditions like $(x,\pi,a)$ or $(x,\textsf{\em lab},\ell)$ are filtered optimally by the Ring, by finding the locus of the path $(\pi,a)$ or $(\textsf{\em lab},\ell)$ in the trie \texttt{POS}: later, when binding $x$, this locus filters the corresponding values. ERing achieves the same by adding, at the time of binding $x$, the leaptor $\textit{leftmost}([c,+\infty],[a,a])$ on $S_\pi$ for properties or the leaptor $select_1(B_\ell,rank_1(B_\ell,c-1)+1)$ for labels.

Conditions that compare properties, like $x.\pi \le y.\pi'$, are translated into $\{ (x,\pi,z), (y,\pi',z'), R(z,z')\}$. In this case, the ERing eliminates (say) $x$ before $y$, and only during the elimination of $y$ it enforces $z \le z'$ (where $z = x.\pi$ is by then a constant). This corresponds to, and is never worse than, eliminating variables in the order $x$, $z$, $y$, $z'$ with the standard Ring.

Since the ERing solution is not worse than the standard Ring solution with some VEO, and it implements all leaps in time $\tilde{O}(1)$, it is wco. \qed
\end{proof}


The ERing can be better than the (standard) Ring (mapped to RDF): it can filter simultaneously what the Ring can only filter sequentially, which is equivalent to pre- or post-filtering. On the other hand, the Ring can use VEOs that are not available to the ERing: consider conveniently translating the clause $x.\pi = y.\pi$ to $\{ (x,\pi,z), (y,\pi,z)\}$, where binding $z$ first might be advantageous.

Conversely, one could represent RDF triples using properties and labels to enable simultaneous filtering over RDF: triples $(u,p,v)$ where $v$ is a literal could map to properties on $u$ ($u.p = v$), while triples $(u,\textsf{\em type},v)$ could map to labels on $u$ ($u.\ell = v$). Our model already supports zero-or-more labels and property values per node, but would require some amendments to cover RDF properties mixing datatype and node values, joins on predicates, and other such details.


\section{Property Graphs} \label{sec:pg}

In the previous section we described techniques for efficiently handling filters on properties and labels. We now apply these techniques for a concrete data model and query language, namely property graphs and a fragment of GQL. Some details of this concrete setting (e.g., edge identifiers, Boolean label expressions) require additional care, while others (e.g., one edge label, at most one value for a property on a given node or edge) provide opportunities to optimize further.

\subsection{Data model and query language} \label{sec:pg-model}

A {\em property graph} is a directed labeled multigraph where any node or edge may feature properties and labels. 
 Let $\G(\N,\E)$ be a property graph with $\N \cap \E = \emptyset$. Let $\PP$ be a set of properties, and $\A$ a set of atomic values properties can take (e.g., integers, dates, strings). Each node or edge has zero or more properties, but a given property has at most one value for a given node or edge. Nodes have zero or more labels in $\LL_N$ and edges have exactly one label in $\LL_E$ ($\LL_N \cap \LL_E = \emptyset$). 
\no{Formally:

\begin{definition}
A property graph is a tuple $\G=(\N, \E, \rho, \lambda,  \sigma)$ where:
\begin{enumerate}
    \item $\rho: \E \to \N \times \N$ is a total function that associates each edge in $\E$ with a source and a target node in $\N$.
    \item $\lambda : (\N \to 2^{\LL_N}) \cup (\E \to 2^{\LL_E})$ is a partial function that associates a node/edge with a set of labels in $\LL_N$/$\LL_E$, with $\LL_N \cap \LL_E = \emptyset$. For every $e \in \E$, $|\lambda(e)|=1$. 
    \item $\sigma : (\N \cup \E) \times \mathcal{P} \to \A$ is a partial function that associates nodes/edges with properties from $\mathcal{P}$, and gives values in $\A$ to those properties.
\end{enumerate}
\end{definition}
}

Our graph of Fig.~\ref{fig:tries} is a property graph, with node property \textsf{age} and edge labels \textsf{\em lives} and \textsf{\em works}. Possible edge properties could be, for example, a date \textsf{since} (for edges with label \textsf{\em lives}) or an integer \textsf{salary} (for edges with label \textsf{\em works}). Possible labels for nodes (representing persons) could be \textsf{\em person} and \textsf{\em has-phd}. 


Two main standards exist to query property graphs: SQL/PGQ and GQL~\cite{francis2023researcher}; their common graph pattern matching language is called GPML~\cite{DeutschFGHLLLMM22}. A GPML query is of the form \texttt{MATCH patterns}, where \texttt{patterns} specifies the subgraph to match as a set of {\em edge patterns}, optionally augmented with a \texttt{WHERE} clause to filter the results. 
We now describe the fragment we support  \cite{DeutschFGHLLLMM22,angles2018property,francis2023researcher}; the formal syntax and semantics are given in  Appendix~\pubext{A~\cite{apps}}{\ref{sec:semantics}}.

The most basic edge patterns are of the form $\mathtt{(s) \to (o)}$, where \texttt{s} and \texttt{o} can be constants in $\N$ or node variables from a set $\V_N$. Those are analogous to triple patterns but do not specify any label. Conditions on labels are associated with node variables, for example
we can add a clause \texttt{(x:person)} to specify that $x \in \V_N$ can match only nodes having the label $\textsf{\em person} \in \LL_N$. \camready{Conditions with no variable names, like \texttt{(:person) $\to$ (y)}, are also permitted.} We can include these clauses in edge patterns, for example in \texttt{(x:person) $\to$ (y)}.
Label conditions may consist of expressions including conjunctions (\texttt{\&}), disjunctions (\texttt{|}), negations (\texttt{!}), and groups of labels. For example, a clause \texttt{(x:person\&has-phd)} admits only the nodes with both labels, \textsf{\em person} and \textsf{\em has-phd}.

Label conditions can also be given on edge patterns, as in $\mathtt{(s) - [\texttt{:works}]}$ $\mathtt{\to (o)}$, which allows matching only edges whose label is $\textsf{\em works} \in \LL_E$; expressions on labels can also be used. It is also possible to give variable names, from a set $\V_E$ disjoint from $\V_N$, to edges: 
$\mathtt{(s) -[e]\rightarrow (o)}$ gives variable name $e \in \V_E$ to the edge; we can also incorporate expressions as in $\mathtt{(s) -[e:works]\rightarrow (o)}$.

The \texttt{WHERE} clause filters the results by restricting the existence or the values of some of their properties or by comparing the identifiers of edges/nodes (yet we cannot compare edge variables with node variables). Conditions can be combined with \texttt{AND}, \texttt{OR}, and \texttt{NOT}. Comparisons are of the form \texttt{(A op B)}, where $\texttt{A}$ and $\texttt{B}$ are constants or variable property values $x.\pi$ of the same type, and $\texttt{op} \in \{ =, \neq, <, \le, \ge, > \}$. We can directly compare node or edge identifiers as if they were properties, as in 
\texttt{(x<=y) AND (y<=z)}, which avoids reporting the same tuples $(x,y,z)$ in every possible order. One can also state that a variable must have, or not have, some property using \texttt{(x.since IS NOT NULL)} or \texttt{(x.since IS NULL)}. 

A final \texttt{RETURN} clause projects the variable or property values to return. We use {\em bag semantics}, so repetitions caused by projection are not removed. 

The following query, for example, looks for names and ages of senior people in Europe working in CS that started in their job before 2010 or have no phd:
\begin{Verbatim}[fontsize=\footnotesize]
   MATCH (x:person)-[:lives]->(Europe), (x)-[e:works]->(CS)
   WHERE (x.age >= 50) AND ((e.since < 01/01/2010) OR (!x.has-phd))
   RETURN x.name, x.age
\end{Verbatim}




\subsection{PG-Ring: A succinct index for property graphs} \label{sec:pg-ring}

We build on the ERing described in Section~\ref{sec:ext-ring} to define PG-Ring, which indexes property graphs within succinct space. 

\camready{For compatibility with the ERing, we create new variables where edge patterns omits them; those are naturally projected out by the \texttt{RETURN} clauses.}


Since the edges of property graphs have exactly one label in $\LL_E$, we treat edges $e = s \to o$ with label $\ell$ as triples $(s,\ell,o)$ in order to create the sequences $L_*$. The symbols of $\Ls$ and $\Lo$ then range over $\N$ and those of $\Lp$ range over $\LL_E$. The wavelet tree of $\Lp$ uses $N\log_2|\LL_E|(1+o(1))$ bits, which is succinct.

Storing the additional edge identifiers of property graphs would not be succinct, so we rather exploit the fact that such identifiers are internal by defining the identifier of an edge as its index in the order \texttt{POS}. We choose this index order since all edge identifiers with some label $\ell$ are precisely in the range $[\Ap[\ell]+1,\Ap[\ell+1]]$, and edge patterns of the form $\texttt{(s)} - [:\ell] \to \texttt{(o)}$ are treated as triple patterns with a constant predicate, instead of using a leaptor.

\camready{Label expressions $L$ on edges,} $\texttt{(s)} - [:L] \to \texttt{(o)}$, are converted into a set of $O(L)$ disjoint ranges: the complement of $\ell$ becomes $\{ [1,\Ap[\ell]],[\Ap[\ell+1]+1,N]\}$, disjunction becomes union of ranges, and conjunction becomes intersection. We then maintain, in general, a set of disjoint ranges during query processing, not just one. For simplicity we consider a single range in the discussion that follows.

\paragraph{Leaptors for edge variables.}

Variables in $\V_E$ can be assigned to edge identifiers. Those will be handled like any other variable in LTJ, with the particularity that their identifiers are not explicit in any sequence $L_*$, but implicit as indices in $\Ls$ (i.e., on order \texttt{POS}). 
The leaptor for a variable $z \in \V_E$ -- which is mentioned on an edge pattern $t$ represented by a trie $\tau_t \in T_z$ -- will thus implement $\leap(c)$ differently depending on which sequence $L_*$ holds the range that represents the locus $v_t$. See Appendix~\pubext{B~\cite{apps}}{\ref{app:pg-ring}} for details.

\paragraph{Leaptors for node properties and labels.}

\camready{Some simplifications with respect to the treatment in Section~\ref{sec:ext-ring} are possible (see Appendix~\pubext{B~\cite{apps}}{\ref{app:pg-ring}}).
For label expressions on nodes, we can handle the complement of a single label $\ell$ with the leaptor for not having a label we described in Section~\ref{sec:ext-ring}. A leaptor for a label expression $L$ builds a syntax tree where negations are pushed down to the leaves and {\em leap$(c)$} is solved by recursively calling {\em leap$(c)$} on both children of internal nodes, taking the minimum returned value $c'$ on disjunctions. For conjunctions, we create multiary nodes containing all the subexpressions that descend by conjunction nodes, and intersect all children as a set, just as for our intersection algorithm. This  matches the alternation complexity of the problem with an extra penalty factor of $O(|L| \log N) \subseteq \tilde{O}(1)$. Bitvectors $B_\pi$ of properties are used in the same way as bitvectors $B_\ell$ of labels, to filter nodes having or not having property $\pi$.}

\paragraph{Boolean leaptors.}

\camready{We use a similar syntax tree, with multiary conjunction nodes, to process together the \texttt{MATCH} edge patterns and the \texttt{WHERE} clauses. This generalizes our set intersections by allowing disjunction nodes (to represent \texttt{OR} in \texttt{WHERE} clauses), which leap by taking the least of their children's leap values. Negations are pushed to the leaves and applied to the atomic condition, both for properties (e.g., \texttt{NOT (x < y)} becomes \texttt{(x >= y)}) and for labels (e.g. \texttt{NOT (x.person)} becomes \texttt{(!x.person)}). See Fig.~\ref{fig:syntax}, and Appendix~\pubext{B~\cite{apps}}{\ref{app:pg-ring}} for details.}

\begin{figure}[t]
\includegraphics[width=\textwidth]{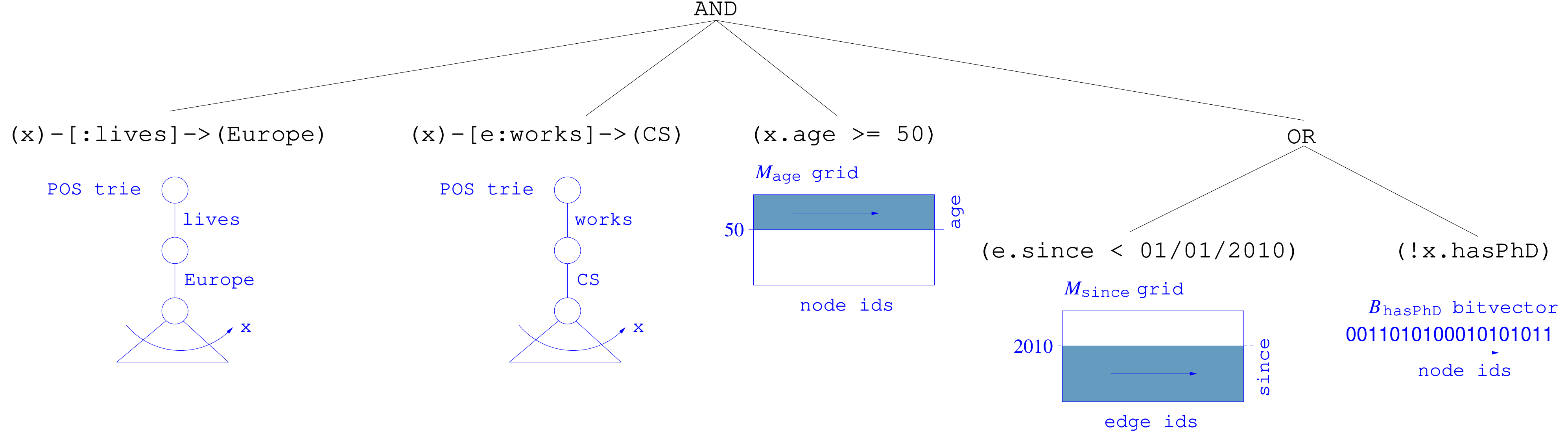}
\caption{The syntax tree of our example of Section~\ref{sec:pg-model}. Each internal node implements a Boolean leaptor; leaf nodes implement trie, grid, and bitvector leaptors.}
\label{fig:syntax}
\end{figure}

\section{Experiments}\label{sec:exp}

We implement PG-Ring in C++11 on top of the succinct data structures library, {\em SDSL}, as described in more detail in Appendix~\pubext{C~\cite{apps}}{\ref{sec:implem}}. We compare it with various state-of-the-art alternatives that support property graphs and GQL-like languages for two different benchmarks: LSQB and Wikidata. We run a set of queries with a time limit of 600 seconds. Following the LSQB benchmark~\cite{mhedhbi2021lsqb}, for each query, we report the average runtime over 5 runs. That measurement is computed after a first warm-up to mitigate the advantage of in-memory systems. Disk activity was monitored using \texttt{iostat} at one-second intervals. In LSQB, a single pass over the queries was enough to avoid disk reads, whereas in Wikidata, additional warm-up queries (e.g., reporting all the nodes) were required beforehand. For both benchmarks, the effect on time performance of a {\em cold cache} scenario is reported in Appendix~\pubext{G~\cite{apps}}{\ref{sec:exp-cold}}. All experiments were conducted on an AMD Ryzen 9 9900X at 4.4GHz, with 24 cores, 32 MB cache, and 249 GB RAM. 

\subsection{Labelled Subgraph Query Benchmark}

We used the Labelled Subgraph Query Benchmark (LSQB)~\cite{mhedhbi2021lsqb} to generate a graph with 35.5 million nodes and 219.4 million edges. Each node can be given one of 14 possible labels, whereas each edge can be given one of 15 possible labels. In this graph, nodes and edges do not have properties. 
We use the first six queries from the benchmark; the others are discarded because they require unsupported optional edges and negation.
See Appendix~\pubext{D~\cite{apps}}{\ref{sec:app-lsqb}} for details.

We compare systems supporting GQL (details in Appendix~\pubext{E~\cite{apps}}{\ref{app:systems}}): Umbra~\cite{NF20},  DuckDB~\cite{duckdb}, K\`uzu~\cite{JFCLS23}, Neo4j~\cite{neo4j}, Memgraph~\cite{memgraph},  MillDB~\cite{VR+23}, and our PG-Ring.
Our figures show the space required by their data and indexes, not the extra needed for query resolution. This extra space is negligible for PG-Ring but not for others that produce intermediate results (we discuss an extreme case soon).
We run the six queries on all the systems, limiting results to $1$, $10$, $100$, $1{,}000$, and $10{,}000$. Appendix~\pubext{D~\cite{apps}}{\ref{sec:app-lsqb}} shows all results; Fig.~\ref{fig:lsqb} summarizes some.  

\begin{figure}[t]
    \centering    \includegraphics[width=0.95\linewidth]{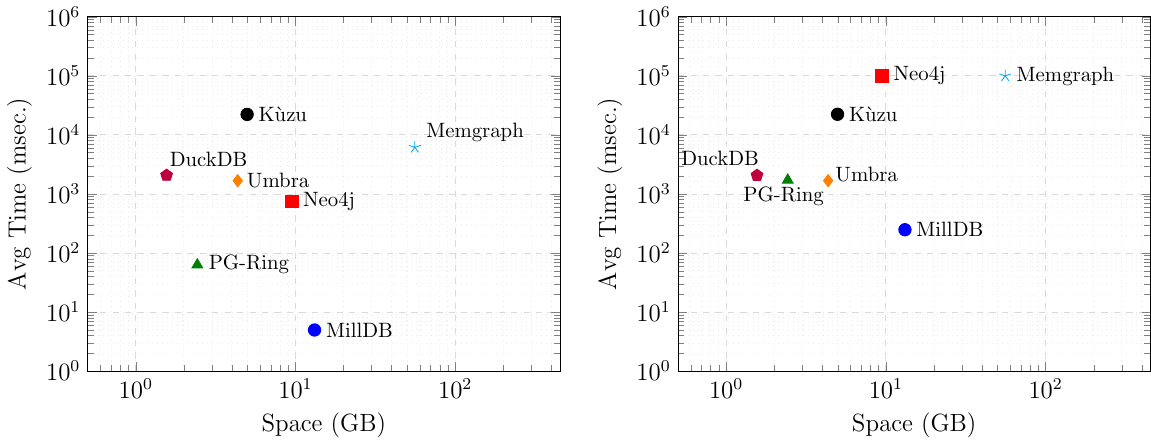}
    \caption{The space--time trade-off of the evaluated systems on six queries of LSQB. The first plot is limited to 10 results, and the second to 1{,}000.}
    \label{fig:lsqb}
\end{figure}

Fig.~\ref{fig:lsqb} (left) shows the space and average time of each system with the number of results limited to 10. We can see that PG-Ring achieves a good balance between space and time. In space, PG-Ring is dominated only by DuckDB, which is $1.6$ times smaller due to its columnar compression, but its average times are $12{-}2{,}700$ times slower (see Appendix~\pubext{D~\cite{apps}}{\ref{sec:app-lsqb}}). The next competitor in space is Umbra, which uses $1.8$ times more space than PG-Ring. In time performance, the PG-Ring wins in three cyclic queries (Q2, Q3, and Q6, see Appendix~\pubext{D~\cite{apps}}{\ref{sec:app-lsqb}}), showing the effect of using the LTJ algorithm to solve these queries. In the other queries, the clear winner is MillDB. In fact, MillDB is very competitive in all queries, but it requires around $5.4$ times more space than PG-Ring.

Fig.~\ref{fig:lsqb} (right) shows the same evaluation when restricting the results to 1{,}000. In this scenario, PG-Ring performs competitively on cyclic queries, with the exception of Q2, where Neo4j performs best. In contrast, Neo4j shows weaker performance in Q3, where it exceeds the time limit, and Q6, where it is $11$ times slower. In Q1, Memgraph stands out by being $1.8$ times faster than MillDB, but it is the most space demanding structure and times out in Q3. For all remaining queries, MillDB is the top performer, taking less than $4$ milliseconds per query. 

The plot also shows that DuckDB and Umbra fare better for obtaining many results. These systems use hash-joins, which require processing the left-side of the join completely and indexing it for lookup before processing the join's right-side (looking up each such value). Even with an output limit, all joins except those on the right extreme of the plan are computed completely, where the last, rightmost table-scan (and the joins above it) can halt early once the limit is reached. Thus the number of intermediate results they process varies little with the limit. The consequence is that they may generate huge intermediate results, preventing their use at larger scale. For example, DuckDB, although using less space than PG-Ring in the static index, generated one billion intermediate results on the cyclic query Q3, which has only 15 million final results. Even on the relatively small LSQB dataset, it ran out of memory on Q1 and Q6 without limit.

In conclusion, PG-Ring offers a good space-time trade-off, it is highly competitive for cyclic queries, highlighting the impact of the wco join algorithm; and ranks second in space usage thanks to the use of compact data structures. 

\subsection{Wikidata Benchmark}

We use the Wikidata knowledge graph~\cite{VrandecicK14} to evaluate our performance in a large scale graph with properties (via qualifiers). The Wikidata graph is represented using the RDF model; therefore, we transformed it into a property graph. The resulting graph consists of $109$ million nodes, which can take  $103$ thousand distinct labels and have $10{,}569$ different properties. In total, there are $719.5$ million edges, with $1{,}640$ possible edge labels and $603$ distinct properties. From the Wikidata SPARQL query logs, we extracted a set of $19.6$ thousand queries and converted them to GQL syntax. See Appendix~\pubext{F~\cite{apps}}{\ref{app:wikidata}} for details, which includes discussions on the limitations of property graphs for representing Wikidata~\cite{HernandezHRRZ16}, in particular for its multi-valued and node-valued qualifiers/properties.

\begin{figure}[t]
    \centering    \includegraphics[width=0.95\linewidth]{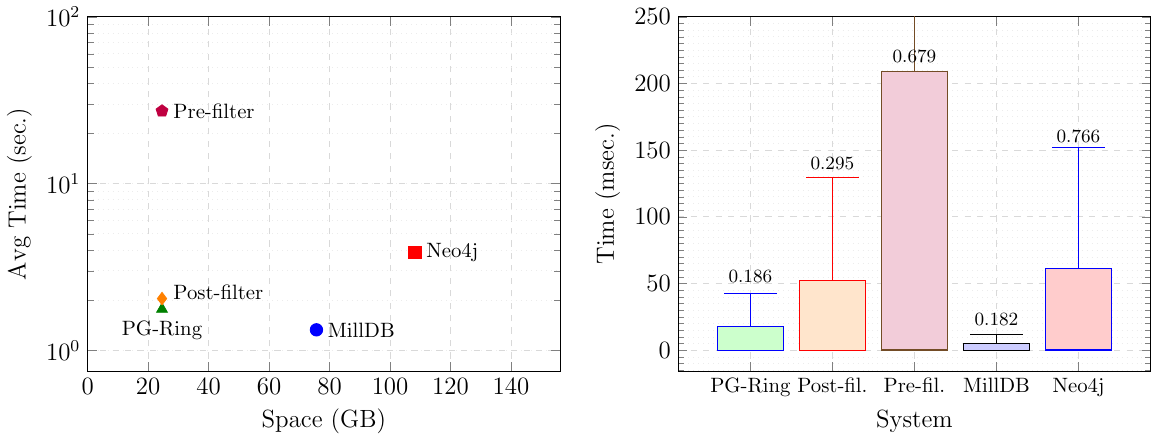}
    \caption{On the left, the space--time trade-off of the systems evaluated on Wikidata. On the right, boxplots represent the distribution of the query execution times. The numbers on top indicate the median time for each system.}
    \label{fig:wiki}
\end{figure}

We exclude Umbra~\cite{NF20}, DuckDB~\cite{duckdb}, and K\`uzu~\cite{JFCLS23} because Wikidata nodes can have multiple labels (classes), but these systems assume a single type/table per entity. Memgraph~\cite{memgraph} was also excluded due to its high memory usage and the impracticality of loading the Wikidata graph. In Neo4j~\cite{neo4j}, we create indexes on the entity identifiers and on the properties used in the queries. 

To study the effect of incorporating the filters directly to the LTJ algorithm, we add two more variants of PG-Ring: The {\em pre-filter} baseline first applies the property and label filters, and only then evaluates the edge patterns. The {\em post-filter} baseline evaluates the edge patterns first and subsequently applies filters.

Fig.~\ref{fig:wiki} shows the performance in space and time of different systems on the selected queries, limiting the output to $1{,}000{,}000$ results. PG-Ring achieves the most favorable space-time trade-off: it requires only $32\%$ of the storage used by MillDB, which is on average just $32\%$ faster than PG-Ring. MillDB shows, however, more stability across all queries, with its third quartile being at around $5$ milliseconds, while the other systems exceed $10$ milliseconds. Conversely, Neo4j is slower than PG-Ring and consumes significantly more space, exceeding $100$~GB: in a real-world scenario Neo4j will consume much more space because we are restricting the indexes to just those properties involved in these specific queries.

The experiment also shows that pushing the filters directly into the LTJ algorithm leads to superior performance. On average, PG-Ring achieves a $16\%$ improvement over the {\em post-filter} baseline and is 15 times faster than {\em pre-filter}. 
The boxplots show a more marked difference in distributions, for example PG-Ring is almost 60\% faster than {\em post-filter} in the median, and much more stable in general. 
As an example to explain this difference, consider one of the queries: 
\begin{Verbatim}[fontsize=\footnotesize]
   MATCH (v:Q5)-[e:P21]->(u) 
   WHERE (v.P569 >= 1554-01-01) AND (v.P569 < 1555-01-01)
   RETURN v, e, u
\end{Verbatim}
which gets {\em the gender (\texttt{P21}) of humans (\texttt{Q5}) born in (\texttt{P569}) 1554}, for which PG-Ring achieves a speedup over $5{,}000$ times compared to the {\em post-filter} baseline, which needs to check for every human if they were born in 1554. Conversely, the {\em pre-filter} is twice as slow as PG-Ring because it first retrieves all entities born in 1554 and only then verifies whether they are humans. Note that Wikidata can also store birth dates for non-human entities (e.g., fictional characters, animals).


\section{Conclusions}\label{sec:conclude}

We proposed succinct data structures that support wco joins and filters. We show the advantages of pushing the filtering within the wco processing of the joins, thereby uplifting the superpowers of those join processing algorithms, and outperforming the classic pre- and post-filtering methods that handle the wco joins as a black box. We adapt and optimize these structures for the concrete setting of evaluating the corresponding fragment of GQL queries over property graphs, giving rise to PG-Ring. Experiments on two property graphs show that although existing systems sometimes outperform PG-Ring in terms of time (MillDB), or space (DuckDB), PG-Ring provides a strong time--space trade-off in this setting.

For future work, we plan to implement PG-Ring for RDF/SPARQL, optimize it specifically for this setting, and compare it with leading systems. Though the Ring \cite{arroyuelo2024ring} supports this setting, an interesting question is when it is preferable to represent (e.g.) datatype values as property values versus graph nodes. 

We only support fragments of languages like GQL and SPARQL, missing typical relational operators (\texttt{UNION}, \texttt{OPTIONAL}), as well as path queries. It would be of interest to explore pushing further query operators from GQL, SPARQL, etc., into low-level operations on the succinct data structures; for example, existing methods for evaluating paths on the Ring could be deployed~\cite{ArroyueloGH0R24}.

A limitation of PG-Ring is that it is currently read-only, but it can be adapted to support updates. The original Ring article \cite[Sec.~8 \& App.~E]{arroyuelo2024ring} discusses in detail how to handle changes on nodes and edges. The PG-Ring represents property values and labels via wavelet trees and bitvectors, for which dynamic versions exist \cite{MNtalg08} while retaining succinctness and enabling updates in time $O(\log^2 N)$, but at the cost of an $O(\log N)$ time penalty factor on queries. This can be alleviated with adaptive dynamic bitvectors~\cite{Navspe25}, a recent development that supports dynamism with a much smaller penalty when the update-to-query ratio is very low, something proven to work well for Wikidata-style workloads~\cite{ACGLNRVvldbj25}. 

\paragraph*{Supplemental Material Statement:} Source code and query sets are available from Github at \url{https://github.com/adriangbrandon/PG-Ring}. The datasets are available from Zenodo at \url{https://zenodo.org/records/19818050}. Appendices are published on arXiv, at \url{https://arxiv.org/abs/2608.03840} \cite{apps}.

\paragraph*{Declaration of use of Generative AI:}

Generative AI was used only as a support tool for code completion, for assisting with the parsing and preprocessing of some data, and to produce Figs.~\ref{fig:tries}--\ref{fig:ring} in tikz from textual descriptions. All suggestions were reviewed, tested, and adapted before being included in the final work.

\bibliographystyle{plain}
\bibliography{biblio}

\pubext{\end{document}}{}
\newpage
\appendix 

\section{Formal Syntax and Semantics of our GQL Fragment} \label{sec:semantics}

\subsection{Property graph}

A property graph is a tuple $\G=(\N, \E, \rho, \LL_N, \LL_E, \lambda, \PP, \A)$ where:
\begin{enumerate}
    \item $\N$ is the set of nodes and $\E$ is a set of directed edge identifiers, $\N \cap \E = \emptyset$.
    \item $\rho: \E \to \N \times \N$ is a total function that associates each edge identifier in $\E$ with its source and target nodes in $\N$, in that order.
    \item $\LL_N$ and $\LL_E$ are sets of node and edge labels, respectively, with $\LL_N \cap \LL_E = \emptyset$.
    \item $\lambda : (\N \to 2^{\LL_N}) \cup (\E \to 2^{\LL_E})$ is a total function that associates nodes with sets of labels in $\LL_N$ and edges with sets of labels in $\LL_E$. It must hold that $|\lambda(e)|=1$ if $e \in \E$, that is, edges have exactly one label. 
    \item $\PP$ is a set of properties and $\A$ is a universe of values (literals). Each $\pi \in \PP$ is a partial function $\pi : \N \cup \E \to \A$.
\end{enumerate}

\subsection{Query syntax}

Queries use variables drawn from sets $\V_N$ (to be matched to nodes) and $\V_E$ (to be matched to edges). Those are disjoint from nodes, edges, and among them: $\V_N \cap (\N \cup \E \cup \V_E) = \emptyset$ and $\V_E \cap (\N \cup \E \cup \V_N) = \emptyset$.
Queries have the form
\begin{Verbatim}[fontsize=\footnotesize]
   MATCH patterns 
   WHERE clauses
   RETURN outputs
\end{Verbatim}
where the \texttt{WHERE} part is optional.
The patterns describe the subgraph to match, as a list of conditioned edges:
{\footnotesize
\begin{eqnarray*}
\texttt{patterns} & \longrightarrow & \texttt{pattern}^+ \\
\texttt{pattern}  & \longrightarrow & \texttt{(node-expr) -> (node-expr)} \\
\texttt{pattern}  & \longrightarrow & \texttt{(node-expr) -[edge-expr]-> (node-expr)} \\
\texttt{node-expr} & \longrightarrow & \texttt{node-base} ~|~ \texttt{: label-expr}(\LL_N) ~|~ \texttt{node-cond} \\
\texttt{node-base} & \longrightarrow & \texttt{node} \in \N ~|~ \texttt{node-var} \in \V_N \\
\texttt{node-cond} & \longrightarrow & \texttt{node-base : label-expr}(\LL_N) \\
\texttt{edge-expr} & \longrightarrow & \texttt{edge-base} ~|~ \texttt{: label-expr}(\LL_E) ~|~ \texttt{edge-cond} \\
\texttt{edge-base} & \longrightarrow & \texttt{edge-var} \in \V_E \\
\texttt{edge-cond} & \longrightarrow & \texttt{edge-base : label-expr}(\LL_E) \\
\texttt{label-expr}(\LL) & \longrightarrow & (\texttt{label} \in \LL)^+ \\
\texttt{label-expr}(\LL) & \longrightarrow & \texttt{label-expr}(\LL) ~\texttt{\&}~ \texttt{label-expr}(\LL) \\
\texttt{label-expr}(\LL) & \longrightarrow & \texttt{label-expr}(\LL) ~\texttt{|}~ \texttt{label-expr}(\LL) \\
\texttt{label-expr}(\LL) & \longrightarrow & \texttt{! label-expr}(\LL) \\
\texttt{label-expr}(\LL) & \longrightarrow & \texttt{( label-expr}(\LL)\texttt{ )} 
\end{eqnarray*}}
%
The clauses set additional conditions on the subgraph matches:
{\footnotesize
\begin{eqnarray*}
\texttt{clauses} & \longrightarrow & \texttt{clause}^+ \\
\texttt{clause}  & \longrightarrow & \texttt{node-cond} ~|~ \texttt{edge-cond} \\
\texttt{clause}  & \longrightarrow & \texttt{node-base op node-base} ~|~ \texttt{edge-base op edge-base} \\
\texttt{clause}  & \longrightarrow & \texttt{node-base.}\pi \in \PP \texttt{ IS NULL} ~|~ \texttt{node-base.}\pi \in \PP \texttt{ IS NOT NULL} \\
\texttt{clause}  & \longrightarrow & \texttt{edge-base.}\pi \in \PP \texttt{ IS NULL} ~|~ \texttt{edge-base.}\pi \in \PP \texttt{ IS NOT NULL} \\
\texttt{clause}  & \longrightarrow & \texttt{node-prop op node-prop} ~|~ \texttt{edge-prop op edge-prop} \\
\texttt{clause}  & \longrightarrow & \texttt{clause AND clause} ~|~ \texttt{clause OR clause} ~|~ \texttt{NOT clause} ~|~ \texttt{( clause )} \\
\texttt{node-prop} & \longrightarrow & a \in \A ~|~ \texttt{node-base.}\pi \in \PP \\
\texttt{edge-prop} & \longrightarrow & a \in \A ~|~ \texttt{edge-base.}\pi \in \PP \\
\texttt{op}        &\longrightarrow & \texttt{=} ~|~ \mathtt{\not}\,\texttt{=} ~|~ \texttt{<} ~|~ \texttt{<=} ~|~ \texttt{>} ~|~ \texttt{>=}
\end{eqnarray*}}
Finally, the outputs indicate what to return from the subgraph matches
{\footnotesize
\begin{eqnarray*}
\texttt{outputs} & \longrightarrow & \texttt{output}^+ \\
\texttt{output}  & \longrightarrow & \texttt{node-base} ~|~ \texttt{edge-base} \\
\texttt{output}  & \longrightarrow & \texttt{node-base.}\pi \in \PP ~|~ \texttt{edge-base.}\pi \in \PP
\end{eqnarray*}}

\subsection{Query semantics}

To define the semantics of a query, we map each edge pattern to a first-order predicate on node and edge variables. Those are not the same variables of the query, but come from another set $\V$. The predicates $N_v$ resulting from nodes and $E_v$ resulting from edges will connect the variables $v \in \V$ with the constants and the variables of the query, from $\V_N$ or $\V_E$.
We translate the edge patterns of the \texttt{MATCH} component into predicates with a function $P$, as follows:
{\footnotesize
\begin{eqnarray*}
P(\texttt{pattern}_1,\ldots,\texttt{pattern}_k) &=& P(\texttt{pattern}_1) \land \cdots \land P(\texttt{pattern}_k) \\
P(\texttt{(node-expr}_1 \texttt{)->(node-expr}_2\texttt{)}) &=& 
\exists e \in \E, u,v \in \N, \rho(e)=(u,v) \\ 
&& \land~ N_u(\texttt{node-expr}_1) \land N_v(\texttt{node-expr}_2) \\
P(\texttt{(node-expr}_1 \texttt{)-[edge-expr]->(node-expr}_2\texttt{)}) &=& 
\exists e \in \E, u,v \in \N, \rho(e)=(u,v) \\ 
&& \land~ N_u(\texttt{node-expr}_1) \land N_v(\texttt{node-expr}_2) \\
&& \land~ E_e(\texttt{edge-expr}) 
\end{eqnarray*}}
\noindent
Functions $N_v$ and $E_v$ map \texttt{node-expr} and \texttt{edge-expr}, respectively, to predicates on the variable $v \in V$. Note that the connection with the constants and variables of the query are made via $N_v(\texttt{node})$, $N_v(\texttt{node-var})$, and $E_v(\texttt{edge-var})$. Similarly, functions $L_v$ map \texttt{label-expr}:

{\footnotesize
\begin{eqnarray*}
N_v(\texttt{node} \in \N) &=& (v = \texttt{node}) \\
N_v(\texttt{node-var} \in \V_N) &=& (v = \texttt{node-var}) \\
N_v(\texttt{: label-expr}) &=& L_v(\texttt{label-expr}) \\
N_v(\texttt{node-base : label-expr}) &=& N_v(\texttt{node-base}) \land L_v(\texttt{label-expr}) \\
E_v(\texttt{edge-var} \in \V_E) &=& (v = \texttt{edge-var}) \\
E_v(\texttt{: label-expr}) &=& L_v(\texttt{label-expr}) \\
E_v(\texttt{edge-base : label-expr}) &=& E_v(\texttt{edge-base}) \land L_v(\texttt{label-expr}) \\
L_v(\texttt{label}_1, \ldots, \texttt{label}_k) &=& (\texttt{label}_1 \in \lambda(v)) \land \cdots \land (\texttt{label}_k \in \lambda(v)) \\
L_v(\texttt{label-expr}_1 \texttt{ \& } \texttt{label-expr}_2) &=& L_v(\texttt{label-expr}_1) \land L_v(\texttt{label-expr}_2) \\
L_v(\texttt{label-expr}_1 \texttt{ | } \texttt{label-expr}_2) &=& L_v(\texttt{label-expr}_1) \lor L_v(\texttt{label-expr}_2) \\
L_v(\texttt{!label-expr}) &=& \neg L_v(\texttt{label-expr}) \\
L_v(\texttt{(label-expr)}) &=& L_v(\texttt{label-expr})
\end{eqnarray*}}
Clauses set further conditions on \texttt{MATCH} via variables in $\V_N \cup \V_E$. Function $C$ connects these conditions by using predicates $N$ and $E$ directly on the variables of $\V_N \cup \V_E$.
{\footnotesize
\begin{eqnarray*}
C(\texttt{clause}_1,\ldots,\texttt{clause}_k) &=& C(\texttt{clause}_1) \land \cdots \land C(\texttt{clause}_k) \\
C(\texttt{node} \in \N \texttt{ : label-expr}) &=& N_{\texttt{node}}(\texttt{label-expr}) \\
C(\texttt{node-var} \in \V_N \texttt{ : label-expr}) &=& N_{\texttt{node-var}}(\texttt{label-expr}) \\
C(\texttt{edge-var} \in \V_E \texttt{ : label-expr}) &=& E_{\texttt{edge-var}}(\texttt{label-expr}) \\
C(\texttt{node-base}_1 \texttt{ op node-base}_2) &=& \texttt{node-base}_1 \texttt{ op } \texttt{node-base}_2 \\
C(\texttt{edge-base}_1 \texttt{ op edge-base}_2) &=& \texttt{edge-base}_1 \texttt{ op } \texttt{edge-base}_2 \\
C(\texttt{node-base.}\pi \in \PP \texttt{ IS}/\texttt{IS NOT NULL}) &=& (\texttt{node-base} \not\in\!/\!\in \textrm{Dom}(\pi)) \\
C(\texttt{edge-base.}\pi \in \PP \texttt{ IS}/\texttt{IS NOT NULL}) &=& (\texttt{edge-base} \not\in\!/\!\in \textrm{Dom}(\pi)) \\
C(\texttt{node-prop}_1 \texttt{ op node-prop}_2) &=& V(\texttt{node-prop}_1) \texttt{ op } V(\texttt{node-prop}_2) \\
C(\texttt{edge-prop}_1 \texttt{ op edge-prop}_2) &=& V(\texttt{edge-prop}_1) \texttt{ op } V(\texttt{edge-prop}_2) \\
C(\texttt{clause}_1 \texttt{ AND clause}_2) &=& C(\texttt{clause}_1) \land C(\texttt{clause}_2) \\
C(\texttt{clause}_1 \texttt{ OR clause}_2) &=& C(\texttt{clause}_1) \lor C(\texttt{clause}_2) \\
C(\texttt{NOT clause}) &=& \neg C(\texttt{clause}) \\
C(\texttt{(clause)}) &=& C(\texttt{clause}) \\
V(a \in \A) &=& a \\
V(\texttt{node-base.}\pi \in \PP) &=& \pi(\texttt{node-base}) \\
V(\texttt{edge-base.}\pi \in \PP) &=& \pi(\texttt{edge-base})
\end{eqnarray*}}
At this point, the expression $P(\texttt{patterns}) \land C(\texttt{clauses})$ is a predicate whose free variables are those of $\V_N \cup \V_E$ used by the query.
We will form, from the \texttt{RETURN} clause, the tuple of values to be returned:
{\footnotesize
\begin{eqnarray*}
O(\texttt{output}_1,\ldots,\texttt{output}_k) &=& (O(\texttt{output}_1),\ldots,O(\texttt{output}_k)) \\
O(\texttt{node-base}) &=& \texttt{node-base} \\ 
O(\texttt{edge-base}) &=& \texttt{edge-base} \\
O(\texttt{node-base.}\pi \in \PP) &=& \pi(\texttt{node-base}) \\
O(\texttt{edge-base.}\pi \in \PP) &=& \pi(\texttt{edge-base})
\end{eqnarray*}}
Finally, the semantics of the \texttt{MATCH-WHERE-RETURN} query is (with bag semantics)
\[ \{ O(\texttt{outputs}) ~\mid~ P(\texttt{patterns}) \land C(\texttt{clauses}) \}. \]

\noindent For illustration, the semantics of the example query given at the end of Section~\ref{sec:pg-model} is as follows:
{\footnotesize
\begin{eqnarray*}
\{ (\texttt{name}(x),\texttt{age}(x)) & \mid & \exists e_1 \in \E, u_1,v_1 \in \N, \rho(e_1)=(u_1,v_1) \\
       &\land& \exists e_2 \in \E, u_2,v_2 \in \N, \rho(e_2)=(u_2,v_2) \\
       &\land& u_1 = x ~\land~ \texttt{person} \in \lambda(u_1) ~\land~ v_1 = \texttt{Europe} ~\land~ \texttt{lives} \in \lambda(e_1) \\
       &\land& u_2=x ~\land~ v_2=\texttt{CS} ~\land~ e_2=e ~\land~ \texttt{works} \in \lambda(e_2) \\
       &\land& \texttt{age}(x) \ge 50 ~\land (\texttt{since}(e) < 01/01/2010 ~\lor~ \neg \texttt{has-phd} \in \lambda(x)) \}
\end{eqnarray*}}

\noindent Here, $x$ and $e$ are free variables used in the query, where the values for two properties of $x$ are projected, while $e$ is not projected.

\section{PG-Ring: Additional Details} \label{app:pg-ring}

%
%
%

\subsection{Leaptors for edge variables}

We implement $\leap(c)$ differently depending on which sequence $L_*$ holds the range that represents the locus $v_t$:

\begin{itemize}
    \item $\Ls[i,j]$: this is the easiest case because the order \texttt{POS} coincides with that of the edge identifiers. We simply return the first element of $[i,j] \cap [c,N]$, that is, $\max(i,c)$ if it is $\le j$, or else $+\infty$.
    \item $\Lo[i,j]$: this case occurs when a value $s$ was previously used to descend from a range $\Ls[i',j']$. We then restrict the range in $\Lo$ by starting from range $[i',j'] \cap [c,N]$ in $\Ls$ and descending by $s$ towards $\Lo[i'',j'']$, as described in Section~\ref{sec:ext-ring} when binding $p$ after $s$. If $i''>j''$ (i.e., the range is empty) we return $+\infty$, otherwise we track $\Lo[i'']$ to $\Ls$ as also described in Section~\ref{sec:ext-ring}, and return the resulting position. 
    \item $\Lp[i,j]$: we first find $p$ such that $\Ap[p] < c \le \Ap[p+1]$, that is, $p$ covers position $c$ in the order \texttt{POS}. Now there are two subcases.
    \begin{enumerate}
        \item The answer is in the range $[c,\Ap[p+1]]$, meaning that its corresponding label is still $p$. To detect this case, we compute the number $k=c-\Ap[p]$ of occurrences of $p$ up to the one in $\Ls[c]=p$, and determine if the $k$th occurrence of $p$, of a posterior one, occurs in $\Lp[i,j]$: we map $\Lp[i,j]$ by $p$ to $\Ls[i',j']$, and if $[i',j'] \cap [c,N]$ is not empty, we return $\max(i',c)$. 
        \item Otherwise, the answer is past the area of label $p$. We then  find the smallest value $p' > p$ in $\Lp[i,j]$ (using operation {\em range next value} on the wavelet tree of $\Lp$). If no such $p'$ exists, we return $+\infty$. Otherwise, we descend by $p'$ towards $\Ls[i'',j'']$ and return $i''$. 
    \end{enumerate}
 \end{itemize}

When it comes to binding the value of an edge variable $z := e$, the resulting range will contain only one entry. We can start from $\Ls[e]$ and track it to the sequence that contains the current range, as explained in Section~\ref{sec:ext-ring}.

\subsection{Leaptors for node properties}

Now that each property $\pi$ defined for a node $v$ has exactly one value $v.\pi$, we can simplify the scheme of Section~\ref{sec:ext-ring}, storing bitvectors $B_\pi[1,|\N|]$ that tell at $B_\pi[v]$ whether node $v$ has property $\pi$ (in the example of Fig.~\ref{fig:grid}, $B_{\textsf{age}}= 01111$). If so, the wavelet tree sequence $S_\pi$ that represents $M_\pi$ has value $S_\pi[rank_1(B_\pi,v)] = v.\pi$, and the first position of $S_\pi$ with node identifier $\ge c$ is $rank_1(B_\pi,c-1)+1$.

A leaptor for $\leap(c)$ with ranges $a^+ \le x.\pi \le b^-$, for constants $a^+$ and $b^-$, is implemented as follows. We first compute $c' = rank_1(B_\pi,c-1)+1$. Second, we compute $c'' \gets \textit{leftmost}([c',+\infty],[a',b'])$ with the wavelet tree of $S_\pi$. Finally, if there is no answer we return $+\infty$ and otherwise we return $select_1(B_\pi,c'')$.

For range comparisons between properties of two variables, say $x.\pi \le y.\pi'$, we proceed as follows for $\leap(c)$ on variable $x$: if $y$ is not yet bound, we compute $c'$ as in the previous paragraph and return $select_1(B_\pi,c')$, the first node $\ge c$ having property $\pi$. Otherwise, we proceed as in the previous paragraph, treating $y.\pi'$ as a constant (which is found at $S_{\pi'}[rank_1(B_{\pi'},y)]$).

We similarly handle comparisons between node identifiers, like $x<y$, treating them in the same way as property comparisons $x.id < y.id$ for a special property $id$, with $x.id = x$ (the numeric identifier of $x \in \N$). That is, when processing $x$, we return $\leap(c) = +\infty$ if $y$ is bound and $c \ge y$, and $\leap(c) = c$ otherwise.

\subsection{Boolean leaptors for the \texttt{MATCH} and \texttt{WHERE} clauses}


We create a single syntax tree with the conditions from the \texttt{MATCH} triples and the \texttt{WHERE} clauses, if any.
Unlike the syntax trees of label expressions, some leaves of this syntax tree may not apply for each current variable $x$ being eliminated, because $x$ is not mentioned in the condition. Similarly, leaves representing edge patterns do not apply if they do not contain the variable $x$ being eliminated. When eliminating $x$, those leaves become {\em trivial} leaptors $\leap(c)=c$. For each new variable $x$ to eliminate, we first simplify the syntax tree for $x$, with the rules $\textit{trivial } \texttt{AND}\, X = X$ and $\textit{trivial } \texttt{OR}\, X = \textit{trivial}$. This, in particular, removes edge patterns that do not participate in the elimination of $x$. In Fig.~\ref{fig:syntax}, for example, the \texttt{OR} node disappears when eliminating $x$. Similarly, bound variables may render conditions directly true or false (such as \texttt{!x.hasPhD} in our example once we eliminate $x$) and those are simplified in the syntax tree for the elimination of further variables: when eliminating variable $e$, the syntax tree is either just the trie for the edge variable $e$ (if $x.\textsf{hasPhD}$ holds for the bound value of $x$) or that node \texttt{AND}ed with the leaf for \texttt{e.since < 01/01/2010} (if not). We say that the reduced syntax tree is {\em reduced for $x$}.

\section{Implementation} \label{sec:implem}

We implement PG-Ring in C++11 on top of the succinct data structures library, {\em SDSL} (\url{https://github.com/simongog/sdsl-lite})~\cite{gbmp2014sea}. The grids of properties $M_{\pi}$ are implemented with wavelet matrices (\texttt{wm\_int}) and $B_{\pi}$ as plain bitvectors (\texttt{bit\_vector}). The bitvectors $B_\ell$ for labels are represented using sparse bitvectors (\texttt{sd\_vector}).

Generally, datasets contain node, edge, and label identifiers as strings, but PG-Ring needs to deal with them as integers. As they are static identifiers, for each of them, we store a dictionary that allows the mapping between strings and integers. We use a dictionary based on Plain-Front Coding from the compressed string dictionaries library, {\em libCSD} (\url{https://github.com/migumar2/libCSD})~\cite{MPBCCNis15}. Similarly, the PG-Ring represents each property value as an integer, so we need to map the original values to an integer depending on the property type (e.g. integers, floats, dates). To this end, we transform each value with a function that associates the original value with an integer identifier. This function is type-dependent and ensures that (i) the original value can be recovered from its identifier, and (ii) the identifier can be compared with other identifiers of the same type in exactly the same way as the original values are compared.

Our system works under {\em bag semantics} (allowing repeated results) and supports {\em homomorphic pattern matching} (allowing two edge patterns match the same graph edge).  When the system receives a query, we use the dictionaries and functions to get the integer identifiers. Once the solutions are obtained, their identifiers are mapped to the original data. We describe next how we choose appropriate variable elimination orders in the presence of clauses on properties and labels. In this prototype, disjunctions in the \texttt{WHERE} clause are not implemented.

\subsection{Variable elimination orders} \label{sec:veo}

The order in which variables are eliminated can have an important impact on practical efficiency \cite{leapfrog}. Good practices for LTJ on graph databases \cite{HRRSiswc19} are (i) leave {\em lonely variables} (those appearing only once in the BGP) to the end, (ii) eliminate first the variables that will produce fewer bindings, and (iii) avoid if possible to bind variables not connected to anyone already bound. For the Ring, it was shown \cite{arroyuelo2024ring} that taking the minimum length of the $L_*$ ranges of all the loci in $T_x$ was a good predictor to choose the variable with fewest bindings. They also showed that {\em adaptive} VEOs, which choose the next variable to eliminate only after binding the current one, and considering the new resulting range lengths (which differ in each branch) outperforms {\em static} VEOs, which define a fixed variable ordering from the beginning.

To choose a good VEO for the PG-Ring, we classify the variables as follows:

\begin{enumerate}
    \item \textit{Non-lonely node variables}: variables in $\V_N$ that appear in two or more edge patterns.
    \item \textit{Non-lonely edge variables}: variables in $\V_E$ that appear in two or more edge patterns.
    \item \textit{Lonely variables}: node or edge variables that occur in just one edge pattern.
\end{enumerate}

The non-lonely node variables are better choices to eliminate first: binding one of them fixes a node across multiple branches, while fixing an edge constrains only one branch. Accordingly, we first bind the non-lonely node variables, then the non-lonely edge variables, and finally the lonely variables. Within each set, we follow policies (ii) and (iii). For (ii), we estimate the {\em selectivity} of the bindings, both for the edge patterns and for the \texttt{WHERE} clauses, to bind first variables with lower selectivity. Selectivity will be the expected fraction of matches for a given variable $x$ and a node of the reduced syntax tree for $x$, assuming uniformity and independence. The selectivity of \texttt{AND} nodes is then the product of the selectivities of their children, and that of an \texttt{OR} node with children selectivities $s_1$ and $s_2$ is $1-(1-s_1)(1-s_2)$. For edge pattern leaves, selectivity is estimated as the length of the range in the corresponding $L_*$ sequence divided by $N$ (i.e., the fraction of edges $x$ can match). For the conditions, we proceed as follows.

\paragraph{Comparisons of property values.} 
Consider a clause comparing $x.\pi$ with $y.\pi'$. We know from the data the ranges $[\min_\pi, \max_\pi]$ for $x.\pi$ and $[\min_{\pi'}, \max_{\pi'}]$ for $y.\pi'$. Selectivity is estimated as the probability that a point from the two-dimensional rectangle $[\min_\pi, \max_\pi]\times[\min_{\pi'}, \max_{\pi'}]$ satisfies the comparison. For example, $x.\pi > y.\pi'$ corresponds to the probability that a point lies above the diagonal line ($x.\pi=y.\pi'$), assuming uniformity. Further, assuming independence, we multiply this probability by the fraction of $1$s in $B_\pi$ and by the fraction of $1$s in $B_{\pi'}$, which are the fractions of node pairs having properties $\pi$ and $\pi'$, respectively. 

In case one of the values is fixed (e.g., $x.\pi > a$ for a constant $a$, or $y$ is bound and $y.\pi'=a$ in $x.\pi > y.\pi'$), the selectivity estimation can be refined to the fraction of the range $[\min_\pi, \max_\pi]$ that satisfies that condition, $(\max_\pi - a)/(\max_\pi - \min_\pi +1 )$, times the fraction of $1$s in $B_\pi$. A more refined option is to  explicitly count the number of points in $M_\pi$ within the range defined by the comparison (in our case, $[1,+\infty] \times [a+1, \infty]$) and divide it by $n$  (where $n=|\N|$ if $x \in \V_N$ and $n=N$ if $x \in \V_E$). This counting takes $O(\log |\A_\pi|)$ time with wavelet trees, which is significant if we have to apply it on every branch. This option is thus applied only on the static VEO (which we implemented and found to be slower than using the dynamic VEO, so we show experiments only on the latter).

\paragraph{Comparisons of identifiers.}

For comparisons involving the identifiers of two variables, the estimation is analogous but easier. As the values of both variables are in the same range $[1, n]$, equality comparisons are assigned a selectivity of $1/n$, inequality comparisons a selectivity of $1-1/n$, and all other types of comparisons a default selectivity of $0.5$. In case one of the identifiers is fixed (e.g., $x > v$ for a constant $v$, or $y$ is bound to $v$ in $x>y$), we estimate it by the fraction of the range $[1, n]$ that satisfies the condition.

\paragraph{Expressions on labels.}

We similarly use the fraction $f_\ell$ of $1$s in $B_\ell$ to estimate the selectivity of a query $x.\ell$. In case of negations, the correct estimation is $1-f_\ell$. For more complex Boolean formulas, we do as with general Boolean expressions.

\paragraph{Internal query optimization.}
Beyond choosing VEOs, we can optimize the reduced Boolean syntax tree that results for each variable $x$ to eliminate: The children of \texttt{AND} nodes are probed left to right, restarting with the leftmost one after each value $c$ in the intersection is found; therefore \texttt{AND} children should be sorted by most selective conditions first (edge patterns and filters can be mixed in this order). Further, range conditions on the same variable $x$ or a property $x.\pi$ should be combined into one (like $x.\pi \ge a$ \texttt{AND} $x.\pi \le b$ into $a \le x.\pi \le b$), because the grid query is more efficient than the alternation complexity. 






\section{The LSQB Benchmark}\label{sec:app-lsqb}

\begin{figure}[t]
    \centering    \includegraphics[width=0.95\linewidth]{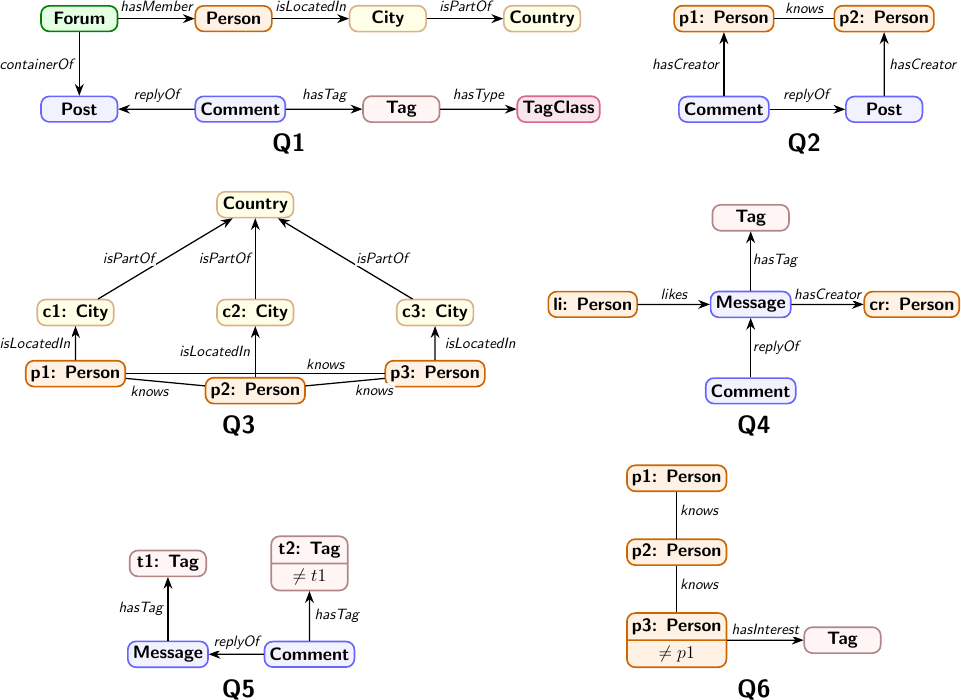}
    \caption{The first six queries from the Labelled Subgraph Queries Benchmark.}
    \label{fig:lsqb-queries}
\end{figure}

The Labelled Subgraph Query Benchmark (LSQB)~\cite{mhedhbi2021lsqb} is a subgraph matching benchmark designed to evaluate the performance of database management systems in graphs with labels in nodes and edges. The benchmark uses the LDBC social network generator~\cite{erling2015ldbc} to create the dataset. In our case, we set the {\em scalability factor} to 10, obtaining 35,496,603 nodes and 219,425,671 edges. Each node can be given one of 14 possible labels, whereas each edge can be given one of 15 possible labels. In this benchmark, nodes and edges do not have properties.

The LSQB benchmark consists of nine queries, where the last three contain optional edges and negations of relationships. The PG-Ring does not support this kind of operators, thus we limit the benchmark to the first six queries.
Fig.~\ref{fig:lsqb-queries} shows the chosen queries and Table~\ref{tab:lsqb-avgs} the average query times obtained.

\begin{table}[p]
    \centering
    \renewcommand{\arraystretch}{1.2}
\setlength{\tabcolsep}{4pt}
\caption{Average times (in msec) of queries, with different limits on the number of results reported. \textcolor{gray}{TO} stands for timeouts, set at 600,000 msec.}
    \label{tab:lsqb-avgs}
\scriptsize
\begin{tabular}{l r r r r r r r}
\toprule
\textbf{System} & \textbf{Limit}
& \textbf{Q1} & \textbf{Q2} & \textbf{Q3} & \textbf{Q4} & \textbf{Q5} & \textbf{Q6} \\
\midrule

\multirow{5}{*}{\textbf{Umbra}}
& 1        & 1{,}912.18 & 462.28 & 159.64 & 5{,}031.74 & 2{,}435.16 & 45.62 \\
& 10       & 1{,}912.38 & 465.54 & 162.26 & 5{,}038.58 & 2{,}432.72 & 45.28 \\
& 100      & 1{,}942.86 & 574.16 & 161.16 & 5{,}061.36 & 2{,}453.38 & 45.68 \\
& 1{,}000  & 1{,}913.06 & 474.24 & 162.04 & 5{,}028.82 & 2{,}437.34 & 46.62 \\
& 10{,}000 & 1{,}939.22 & 509.38 & 188.14 & 5{,}091.54 & 2{,}489.56 & 48.70 \\
\midrule

\multirow{5}{*}{\textbf{DuckDB}}
& 1        & 3{,}111.74 & 454.40 & 300.04 & 5{,}222.78 & 3{,}198.34 & 131.96 \\
& 10       & 3{,}071.52 & 446.38 & 299.24 & 5{,}138.54 & 3{,}199.64 & 137.22 \\
& 100      & 3{,}148.88 & 453.56 & 304.50 & 5{,}236.04 & 3{,}247.80 & 138.80 \\
& 1{,}000  & 3{,}109.04 & 451.96 & 301.62 & 5{,}196.26 & 3{,}202.10 & 138.52 \\
& 10{,}000 & 3{,}128.20 & \textbf{455.14} & 311.40 & 5{,}229.64 & 3{,}222.50 & 140.64 \\
\midrule

\multirow{5}{*}{\textbf{Kuzu}}
& 1        & 11{,}352.10 & 1{,}403.96 & 115{,}401.98 & 2{,}277.40 & 2{,}796.80 & 60.84 \\
& 10       & 11{,}207.38 & 1{,}394.62 & 115{,}999.14 & 2{,}260.88 & 2{,}772.04 & 61.06 \\
& 100      & 11{,}372.08 & 1{,}392.02 & 116{,}026.98 & 2{,}271.80 & 2{,}777.86 & 63.20 \\
& 1{,}000  & 11{,}440.90 & 1{,}406.14 & 116{,}053.84 & 2{,}279.18 & 2{,}784.96 & 65.10 \\
& 10{,}000 & 11{,}201.86 & 1{,}432.66 & 115{,}931.40 & 2{,}291.42 & 2{,}799.98 & 80.84 \\
\midrule

\multirow{5}{*}{\textbf{Neo4j}}
& 1      & 7.56  & 2.60 & 2{,}165.30 & 3.14 & 4.40 & 1.86 \\
& 10     & 8.92  & 3.66 & 4{,}439.06 & 3.76 & 5.82 & 2.26 \\
& 100    & 12.24 & \textbf{6.42} & \timeout   & 5.76 & 16.76 & 5.38 \\
& 1{,}000  & 56.22 & \textbf{32.08} & \timeout & 35.88 & 33.96 & 27.28 \\
& 10{,}000 & 479.56 & 2{,}449.34 & \timeout & 337.00 & 291.10 & 325.70 \\
\midrule

\multirow{5}{*}{\textbf{Memgraph}}
& 1      & \textbf{5.04}  & 42.26 & 36{,}936.34 & 39.92 & 9.06 & 2.56 \\
& 10     & \textbf{4.38} & 37.88 & 65{,}980.02 & 36.84 & 8.00 & 1.72 \\
& 100 & \textbf{5.18} & 38.88 & \timeout & 35.32 & 10.70 & 1.10 \\
& 1{,}000  & \textbf{17.26} & 192.38 & \timeout & 50.88 & 9.46 & 4.50 \\
& 10{,}000 & \textbf{122.20} & 810.90 & \timeout & 86.12 & 52.16 & 36.02 \\
\midrule

\multirow{5}{*}{\textbf{MillDB}}
& 1 & 5.32 & 1.78 & 9.19 & 2.64 & \textbf{1.39} & 1.78 \\
& 10     & 5.79 & 2.25 & 15.80 & \textbf{2.70} & \textbf{1.47} & 1.84 \\
& 100 & 9.16 & 11.79 & 147.86 & 2.00 & \textbf{1.47} & 1.08\\
& 1{,}000  & 31.34 & 77.91 & 1{,}368.47 & \textbf{3.44} & \textbf{3.14} & \textbf{1.93} \\
& 10{,}000 & 165.66 & 786.65 & 13{,}582.08 & \textbf{15.98} & \textbf{21.62} & \textbf{9.59} \\
\midrule

\multirow{5}{*}{\textbf{PG-Ring}}
& 1 & 127.37 & \textbf{0.60} & \textbf{0.12} & 2.07 &103.46 & \textbf{0.03} \\
& 10     & 259.25 & \textbf{1.45} & \textbf{0.66} & 3.24 & 114.01 & \textbf{0.05} \\
& 100 & 1{,}120.93 & 17.00 & \textbf{3.67} & 18.62 & 211.50 & \textbf{0.25} \\
& 1{,}000  & 8{,}810.63 & 173.70 & \textbf{24.17} & 150.87 & 1{,}135.47 & 2.48 \\
& 10{,}000 & 56{,}581.03 & 1{,}601.01 & 275.35 & 902.16 & 10{,}086.07 & 19.99 \\

\bottomrule
\end{tabular}
\end{table}

\section{Systems Compared} \label{app:systems}

In the evaluation, we use these systems with the following settings:
\begin{itemize}
    \item Umbra~\cite{NF20}: A system based on relational tables whose query plans use binary joins but may introduce wco plans for some sub-queries.
    \item DuckDB~\cite{duckdb}: A non-wco relational database that stores data in columnar format. It uses vectorized and pipelined models to produce optimized binary join plans.
    \item K\`uzu~\cite{JFCLS23}: A graph query engine that stores property graphs using unsorted adjacency lists. It employs cost-based dynamic programming to generate execution plans that combine wco and pairwise joins.
    \item Neo4j~\cite{neo4j}: A widely used graph database that implements the property graph model. Its execution plans are based on index lookups, expand operators, and pairwise joins. 
    \item Memgraph~\cite{memgraph}: An in-memory graph database that stores the property graph model using adjacency structures. Query plans rely on traversal-based operators combined with pairwise joins.
    \item MillDB~\cite{VR+23}: A graph database that supports the property graph model using a quad-based storage. The query plans combine wco joins and pairwise joins.
    \item PG-Ring: the system presented in this work.
\end{itemize}

We ensure that all systems run in a single-threaded execution. All the systems use non-repeated edge semantics, except MillDB and PG-Ring. For this reason, for those two systems the queries of LSQB are modified adding in the \texttt{WHERE} clause additional expressions to avoid repeated edges. As K\`uzu, PG-Ring does not support undirected edges in the query pattern. To support all LSQB queries, we handle undirected edges by adding them in both directions.

\section{The Wikidata Benchmark}
\label{app:wikidata}

We convert a Wikidata knowledge graph~\cite{VrandecicK14} into a property graph as follows. 
Each Wikidata entity is modeled as a node, with the exception of those referenced by the \textsf{{\em instance-of}} property, which are converted into node labels. Those literals (e.g. strings, numbers, coordinates, dates) linked to a node are represented as properties of the node. The remaining RDF properties are modeled as relationships whose label is the RDF property itself. The edges store the temporal, numeric, and geospatial properties obtained from their qualifiers. In addition, we verify that the added values match the corresponding property type. The resulting graph consists of $108{,}882{,}974$ nodes, which can take on $102{,}718$ distinct labels and have $10{,}569$ different properties. In total, there are $719{,}656{,}925$ edges, featuring $1{,}640$ possible edge labels and $603$ distinct properties.

It is important to note that, per well-known limitations~\cite{HernandezHRRZ16}, we cannot directly model Wikidata completely with property graphs. Specifically, we cannot directly represent qualifiers that have items/nodes as values (we can add a string ID of the node, but this does not necessarily ``join" with the node, unless we duplicate the ID as a node property) or multi-valued properties (which could be captured as arrays or lists where supported as datatypes). We omit such features for simplicity as they should not affect comparative performance.

We obtained $19{,}605$ queries from the Wikidata SPARQL query log by extracting BGPs with filters and applying a similar transformation to the BGPs (and filters) in order to match the transformed knowledge graph. Queries that are syntactically similar (e.g., same graph pattern, same filter types) are grouped together, regardless of constants and operands. 

More precisely,
from each SPARQL query in the Wikidata logs, we extract (maximal) subqueries corresponding to a BGP with (if present) filters. We skip queries with multiple such BGPS (ignoring SERVICE clauses), and remove BGPs: (1) with Cartesian products (not fully connected via variables), (2) with constants not in the graph, (3) with unsupported filters (beyond $<$, $>$, $=$, $<=$, $\>=$, !, \&, $\mid$) or filters with variables not in the BGP, and (4) seen previously (modulo isomorphism). Per the data mapping, coordinates are split into two lat/long properties, constant values for P31 (instance of) are added as node labels, relations between items are added as edges with a fresh edge variable, datatype property values are added as properties to nodes, and datatype qualifier values are added as properties to edges. The resulting query is written to GQL syntax; if the BGP contains a non-projected variable on a datatype property Pn for subject s, we add ``IS NOT NULL(s.Pn)'' to ensure existence.

These queries are classified into different groups by computing a structural fingerprint of the graph pattern and a descriptor of the WHERE clause. We find 391 groups of queries partitioning by fingerprint. The top 5 are simple triple patterns without WHERE clause, covering almost 50\% of the queries. We select $1{,}000$ queries by sampling uniformly while ensuring at least one representative of each group.

\section{Cold Cache Scenario} \label{sec:exp-cold}

For fairness to disk-based systems, we reported the averaged time in a {\em warm cache} scenario. We have also evaluated the systems in a {\em cold cache} setup: before the execution of each set of queries for each system, we clean the cache with the command \texttt{echo 3 | sudo tee /proc/sys/vm/drop\_caches}. Then, we measure the running times of a single execution of each query. 

Fig.~\ref{fig:wiki-cold} and Table~\ref{tab:wiki-cold} detail the results on both benchmarks in the {\em cold cache} scenario. We can observe how the performance of disk-based systems degrades with respect to the {\em warm cache} scenario. This time, the outstanding systems in the LSQB benchmark are Memgraph and PG-Ring. On the Wikidata benchmark, the cold cache setup sharply affects the performance of both MillDB and Neo4j, as they are disk-based.

\begin{table}
    \centering  
    \renewcommand{\arraystretch}{1.2}
    \setlength{\tabcolsep}{4pt}

    \caption{Detailed execution times (in msec) for the six queries of LSQB in the {\em cold cache} scenario. \textcolor{gray}{TO} stands for timeouts, set at 600,000 msec.}
    \label{tab:wiki-cold}
    
    \scriptsize
    \begin{tabular}{l r r r r r r r}
    \toprule
    \textbf{System} & \textbf{Limit}
    & \textbf{Q1} & \textbf{Q2} & \textbf{Q3} & \textbf{Q4} & \textbf{Q5} & \textbf{Q6} \\
    \midrule
    
    \multirow{5}{*}{\textbf{Umbra}} 
    & 1      & 37{,}701.30 & 2{,}261.50 & 190.40 & 14{,}710.50 & 2{,}512.00 & 421.70 \\
    & 10     & 37{,}806.90 & 2{,}235.14 & 188.86 & 14{,}835.52 & 2{,}547.69 & 426.12 \\
    & 100    & 37{,}858.30 & 2{,}234.60 & 188.70 & 14{,}812.40 & 2{,}504.10 & 417.30 \\
    & 1{,}000  & 37{,}875.42 & 2{,}233.34 & 195.78 & 14{,}842.23 & 2{,}537.86 & 423.26 \\
    & 10{,}000  & 38{,}079.30 & 2{,}277.30 & \textbf{211.60} & 14{,}841.10 & 2{,}520.60 & 426.90 \\
    \midrule
    
    \multirow{5}{*}{\textbf{DuckDB}} 
    & 1      & 12{,}018.00 & 807.70 & 295.90 & 11{,}489.20 & 4{,}370.30 & 187.40 \\
    & 10     & 12{,}321.13 & 884.21 & 297.54 & 11{,}795.84 & 4{,}263.62 & 182.53 \\
    & 100    & 12{,}164.50 & 854.20 & 293.10 & 11{,}637.20 & 4{,}254.40 & 155.80 \\
    & 1{,}000  & 11{,}929.72 & 867.36 & 301.84 & 11{,}665.70 & 4{,}289.34 & 166.20 \\
    & 10{,}000  & 12{,}234.70 & 854.30 & 301.30 & 11{,}430.60 & 4{,}447.40 & 191.30 \\
    \midrule
    
    \multirow{5}{*}{\textbf{Kuzu}} 
    & 1      & 15{,}634.50 & 4{,}772.30 & 116{,}282.50 & 3{,}556.90 & 3{,}742.70 & 243.80 \\
    & 10     & 16{,}416.00 & 3{,}140.02 & 118{,}911.93 & 3{,}314.12 & 3{,}875.07 & 140.09 \\
    & 100    & 16{,}518.40 & 3{,}170.80 & 117{,}082.60 & 3{,}258.20 & 3{,}877.30 & 136.80 \\
    & 1{,}000  & 16{,}404.10 & 3{,}174.60 & 118{,}559.12 & 3{,}358.14 & 3{,}798.43 & 152.60 \\
    & 10{,}000  & 15{,}722.20 & 3{,}235.60 & 115{,}890.00 & 3{,}348.20 & 3{,}788.40 & 172.40 \\
    \midrule
    
    \multirow{5}{*}{\textbf{Neo4j}} 
    & 1     & 1{,}324.90 & 572.60 & 173{,}106.40 & 268.50 & 13{,}181.40 & 145.30 \\
    & 10    & 2{,}316.14 & 760.80 & 170{,}786.27 & 407.12 & 14{,}725.46 & 113.13 \\
    & 100   & 4{,}290.70 & 3{,}185.90 & \timeout & 621.70 & 16{,}108.20 & 224.40 \\
    & 1{,}000  & 6{,}115.30 & 23{,}148.54 & \timeout & 4{,}135.32 & 24{,}139.87 & 442.88 \\
    & 10{,}000  & 7{,}307.90 & 142{,}413.20 & \timeout & 49{,}078.00 & 69{,}088.20 & 313.00 \\
    \midrule
    
    \multirow{5}{*}{\textbf{Memgraph}} 
    & 1     & \textbf{21.50} & 42.00 & 19{,}912.40 & 27.60 & 6.80 & 1.30 \\
    & 10    & \textbf{21.08} & 36.92 & 41{,}272.58 & 24.26 & \textbf{6.42} & 1.64 \\
    & 100   & \textbf{24.60} & 50.50 & \timeout & 34.40 & \textbf{5.60} & 2.00 \\
    & 1{,}000  & \textbf{28.14} & \textbf{113.78} & \timeout & \textbf{27.68} & \textbf{6.64} & \textbf{4.10} \\
    & 10{,}000  & \textbf{99.10} & \textbf{423.80} & \timeout & \textbf{59.70} & \textbf{30.20} & 23.60 \\
    \midrule
    
    \multirow{5}{*}{\textbf{MillDB}} 
    & 1     &  1{,}599.68 & 451.58 & 634.56 & 233.53 & \textbf{2.35} & 214.05 \\
    & 10    & 1{,}751.43 & 995.05 & 687.08 & 273.11 & 7.86 & 157.41 \\
    & 100     & 2{,}273.16 & 5{,}140.34 & 825.91 & 588.99 & 32.13 & 156.70 \\
    & 1{,}000  & 6{,}687.00 & 33{,}693.54 & 2{,}287.22 & 465.04 & 71.40 & 165.09 \\
    & 10{,}000     &  20{,}904.32 & 49{,}924.80 & 14{,}777.50  & 547.92  & 665.80 & 298.50 \\
    \midrule
    
    \multirow{5}{*}{\textbf{PG-Ring}} 
    & 1     & 129.84 & \textbf{1.49} & \textbf{0.28} & \textbf{4.98} & 105.26  & \textbf{0.08} \\
    & 10     & 262.09 & \textbf{3.04} & \textbf{1.22} & \textbf{1.46} & 117.94 & \textbf{0.13} \\
    & 100     & 1{,}131.02  &  \textbf{23.18} & \textbf{5.26} & \textbf{24.55} & 212.99 & \textbf{0.52} \\
    & 1{,}000  & 8{,}804.49 & 180.28 & \textbf{28.82} & 157.19 & 1{,}136.67 & 4.37 \\
    & 10{,}000  &  56{,}307.18 & 1{,}603.29 & 281.60 & 911.71 & 10{,}081.64 & \textbf{23.35} \\
    
    \bottomrule
    \end{tabular}
\end{table}

\begin{figure}[t!]
    \includegraphics[width=0.95\linewidth]{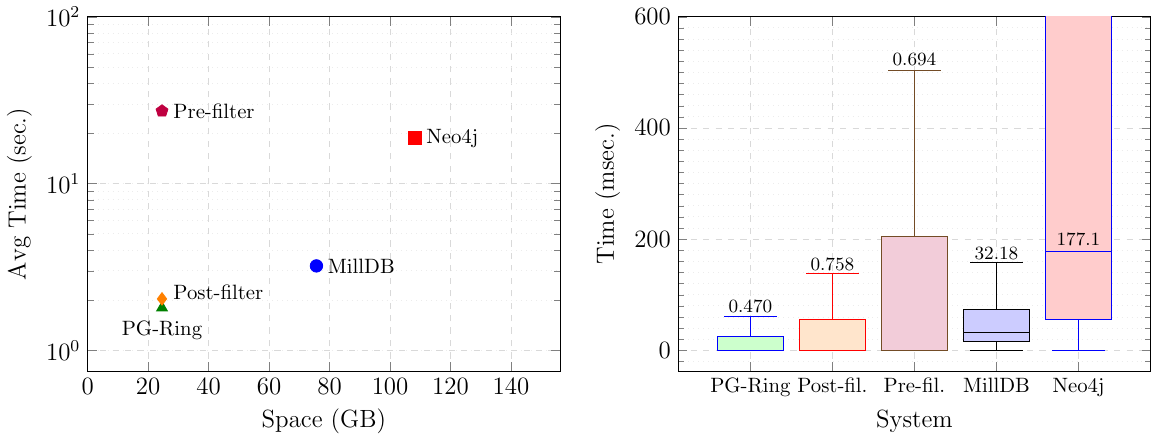}
    \caption{The space-time trade-off and times distribution on the Wikidata benchmark in the \textit{cold cache} scenario. The numbers in the boxplots indicate the median of execution times.}
    \label{fig:wiki-cold}
\end{figure}

\end{document}